\documentclass[11pt]{elsarticle}

\usepackage{amssymb}
\usepackage{amsmath}
\usepackage{hyperref}
\usepackage{tikz}
\usepackage{listings}
\usepackage{pythontex}
\usepackage{amsthm}
\usepackage{amsmath, amsfonts, amssymb, amsthm}

\newtheorem{Definition}{Definition}[section]
\newtheorem{Theorem}[Definition]{Theorem}
\newtheorem{Lemma}[Definition]{Lemma}
\newtheorem{Proposition}[Definition]{Proposition}
\newtheorem{Corollary}[Definition]{Corollary}
\newtheorem{Example}[Definition]{Example}
\newtheorem{Remark}[Definition]{Remark}

\usepackage{array,multirow,makecell}

\setcellgapes{1pt} \makegapedcells
\newcolumntype{R}[1]{\rangle{\raggedleft\arraybackslash }b{#1}}
\newcolumntype{L}[1]{\rangle{\raggedright\arraybackslash }b{#1}}
\newcolumntype{C}[1]{\rangle{\centering\arraybackslash }b{#1}}
\newcounter{minutes}
\divide\time by 60
\newcounter{hours}
\multiply\time by 60 \addtocounter{minutes}{-\time}

\definecolor{codeblue}{rgb}{0.25, 0.5, 0.75}
\definecolor{codegreen}{rgb}{0, 0.6, 0}
\definecolor{codegray}{rgb}{0.5, 0.5, 0.5}
\definecolor{codepurple}{rgb}{0.58, 0, 0.82}
\definecolor{backcolour}{rgb}{0.95, 0.95, 0.92}

\makeatletter
\def\ps@pprintTitle{%
 \let\@oddhead\@empty
 \let\@evenhead\@empty
 \let\@oddfoot\@empty
 \let\@evenfoot\@oddfoot
} \makeatother

\begin{document}

\begin{frontmatter}

\title{ Cyclic Codes of Length $7p^s$ over $\mathbb{F}_{p^m} + u\mathbb{F}_{p^m}$: Characterization, Duals, and Applications to Quantum and LCD Codes }

\author[VB]{Payel Chandra}
\ead{payelchandra12@gmail.com}

\author[VB]{Kalyan Hansda \corref{cor1}}
\ead{kalyanh4@gmail.com}

\cortext[cor1]{Corresponding author.}

\address[VB]{Department of Mathematics, Visva-Bharati University, Santiniketan,
Bolpur - 731235, West Bengal, India.}

\begin{abstract}
Let $R_2 = \mathbb{F}_{p^m} + u\mathbb{F}_{p^m}$ ($u^2 = 0$), where $p$ is an odd prime and $m, s \in \mathbb{N}$. For $p \equiv 3, 5 \pmod 7$ with $\gcd(m, 6) = 1$, the cyclotomic polynomial $\Phi_7(x)$ is irreducible over $\mathbb{F}_{p^m}$. This yields a direct sum decomposition $C = C_1 \oplus C_2$ for any cyclic code $C$ of length $7p^s$ over $R_2$, where $C_1$ has length $p^s$ and $C_2$ is a $7$-cyclotomic code of length $6p^s$. We classify $7$-cyclotomic codes into four disjoint generator-based types and calculate exact cardinalities using residue and torsion subcodes. Furthermore, explicit generators for the Euclidean dual codes $C^\perp$ are determined. As operational applications of these classified codes, we construct new families of quantum stabilizer codes via the CSS framework and establish parameter criteria for linear codes with complementary duals (LCD codes).
\end{abstract}

\end{frontmatter}
\section{Introduction}

Cyclic codes form one of the most fundamental classes of linear error-correcting codes due to their rich algebraic structure and efficient shift-register implementations \cite{Prange1957, Prange1958}. While cyclic codes were originally developed over finite fields, their study over finite commutative rings gained significant momentum following the discovery that key non-linear binary codes map directly from linear codes over $\mathbb{Z}_4$ \cite{DinhLopez2004}. Among finite commutative ring structures, local chain rings—particularly $R_q = \mathbb{F}_{p^m} + u\mathbb{F}_{p^m}+u^2\mathbb{F}_{p^m}+...+u^{q-1}\mathbb{F}_{p^m} = \mathbb{F}_{p^m}[u]/\langle u^q \rangle$ with $u^q = 0$—have become a benchmark alphabet for constructing optimal ring-linear codes.

A cyclic code of length $n$ over $R_2$ corresponds to an ideal of the quotient ring $R_2[x]/\langle x^n - 1 \rangle$ \cite{DinhLopez2004}. When $p \nmid n$, the polynomial $x^n - 1$ factors into pairwise coprime irreducible factors, yielding a semisimple quotient ring generated by idempotents. Conversely, when $p \mid n$, $x^n - 1$ contains repeated factors over $R_2$, giving rise to \textit{repeated-root cyclic codes}. The algebraic structure of repeated-root codes is substantially more intricate due to the presence of non-trivial nilpotent elements and non-principal ideal chains.

In recent years, classifying repeated-root cyclic and constacyclic codes over finite chain rings has attracted widespread attention. Dinh established the structural classification and dual codes for length $p^s$ over $R_2$ \cite{Dinh2010}. Subsequent works extended this classification framework to length $2p^s$ \cite{Dinh2012, Liu2014}, length $3p^s$ \cite{Dinh2013a, Phuto2020}, length $4p^s$ \cite{Dinh2013b}, length $5p^s$ \cite{Boudine2022}, and general prime-multiple length $\ell p^s$ \cite{Chen2014}, alongside extensions to distinct ring structures \cite{Kiah2008, Kiah2012, Sobhani2015}. However, an explicit classification, exact cardinality evaluation, and dual code derivation for cyclic codes of length $7p^s$ over $R_2$ remained open.

The factorization of $x^{7p^s} - 1$ over $\mathbb{F}_{p^m}$ depends on the factorability of the $7$-th cyclotomic polynomial $\Phi_7(x) = x^6 + x^5 + x^4 + x^3 + x^2 + x + 1$ \cite{Arnold2011, Wu2016}. Under specific arithmetic conditions ($p \equiv 3, 5 \pmod 7$ and $\gcd(m, 6) = 1$), $\Phi_7(x)$ remains irreducible over $\mathbb{F}_{p^m}$. By applying the Chinese Remainder Theorem, the ambient ring ring algebra decomposes as \cite{Atiyah2018}:
$$\frac{R_2[x]}{\langle x^{7p^s} - 1 \rangle} \cong \frac{R_2[x]}{\langle (x - 1)^{p^s} \rangle} \oplus \frac{R_2[x]}{\langle (\Phi_7(x))^{p^s} \rangle}.$$

Consequently, any cyclic code $C$ of length $7p^s$ over $R_2$ splits uniquely into a direct sum $C = C_1 \oplus C_2$, where $C_1$ is a cyclic code of length $p^s$ and $C_2$ is a $7$-cyclotomic code of length $6p^s$ over $R_2$.

The main contributions of this paper are fourfold. First, we establish a complete structural classification of $7$-cyclotomic codes of length $6p^s$ over $R_2$ into four disjoint generator-based types. Second, using residue and torsion code techniques derived from the rank-nullity theorem, we derive exact counting formulas for all classified code types and their direct sum representations. Third, we explicitly compute the generator polynomials for the Euclidean dual codes $C^\perp$ of all $7$-cyclotomic and $7p^s$-length cyclic codes over $R_2$. Finally, we demonstrate two concrete applications by constructing new families of Quantum Error-Correcting Codes (QECCs) via the CSS construction and characterizing Linear Codes with Complementary Duals (LCD codes).

The rest of this paper is organized as follows. Section~2 reviews algebraic preliminaries regarding finite chain rings, cyclotomic polynomials, and Gray maps. Section~3 details the direct sum decomposition, classification of $7$-cyclotomic codes, codeword enumeration, and dual code derivations. Section~4 presents the applications to quantum error-correction with examples and LCD codes. 

\section{Preliminaries}

In this section, we review essential algebraic preliminaries regarding finite chain rings, cyclotomic polynomials, repeated-root cyclic codes of length $p^s$ over $R_2$, and the canonical Gray map. In addition, we introduce the foundational concepts of Quantum Error-Correcting Codes (QECCs) and Linear Codes with Complementary Duals (LCD Codes).

\subsection{Ring Structure and Cyclotomic Polynomials}

Throughout this paper, $p$ denotes an odd prime, $m, s \in \mathbb{N}$, and $\mathbb{F}_{p^m}$ is the finite field with $p^m$ elements. We consider the local commutative chain ring

$$R_2 = \frac{\mathbb{F}_{p^m}[u]}{\langle u^2 \rangle} = \{x = x_0 + u x_1 \mid x_0, x_1 \in \mathbb{F}_{p^m}, \, u^2 = 0\}.$$

Every element $x \in R_2$ can be uniquely expressed as $x = x_0 + u x_1$ for $x_0, x_1 \in \mathbb{F}_{p^m}$. An element $x = x_0 + u x_1$ is a unit in $R_2$ if and only if  $x_0 \in \mathbb{F}_{p^m}^\times$. The ideal structure of $R_2$ is chain-conditioned, with unique maximal ideal $\mathfrak{m} = \langle u \rangle = u\mathbb{F}_{p^m}$, making it a principal ideal ring with residue field $R_2/\mathfrak{m} \cong \mathbb{F}_{p^m}$. The unit group of $R_2$ is denoted by $R_2^\times$.

\begin{Definition}[\cite{Burton2010}]
An integer $g$ is a \textbf{primitive root modulo $n$} if the residue class of $g$ generates the multiplicative group of units modulo $n$, i.e., $\mathbb{Z}_n^\times = \langle g \pmod n \rangle$. Equivalently, $g$ is coprime to $n$ and has multiplicative order $\phi(n)$ modulo $n$, where $\phi$ denotes Euler's totient function.
\end{Definition}

For $n = 7$, the group of units $\mathbb{Z}_7^\times = \{1, 2, 3, 4, 5, 6\}$ is cyclic of order $\phi(7) = 6$. The generators of $\mathbb{Z}_7^\times$ are $3$ and $5 \pmod 7$. For instance, $3^1 \equiv 3$, $3^2 \equiv 2$, $3^3 \equiv 6$, $3^4 \equiv 4$, $3^5 \equiv 5$, and $3^6 \equiv 1 \pmod 7$. Conversely, $4$ is not a primitive root modulo $7$ since $\langle 4 \pmod 7 \rangle = \{1, 2, 4\} \subsetneq \mathbb{Z}_7^\times$.

\begin{Definition}[\cite{Arnold2011}]
For $n \in \mathbb{N}$, the $n$-th \textbf{cyclotomic polynomial} $\Phi_n(x) \in \mathbb{Z}[x]$ is defined as

$$\Phi_n(x) = \prod_{\substack{1 \le j \le n \\ \gcd(j, n) = 1}} (x - \xi^j).$$

where $\xi$ is a primitive $n$-th root of unity in $\mathbb{C}$. The polynomial $x^n - 1$ decomposes into cyclotomic polynomials as $x^n - 1 = \prod_{d \mid n} \Phi_d(x)$. When $n$ is prime, $\Phi_n(x) = \sum_{j=0}^{n-1} x^j = x^{n-1} + x^{n-2} + \dots + x + 1$.
\end{Definition}

\begin{Definition}[\cite{Boudine2022}]
Let $R$ be a finite commutative ring with unity. An $n$-\textbf{cyclotomic code} of length $d_n k$ over $R$ is defined as an ideal of the quotient ring $R[x]/\langle (\Phi_n(x))^k \rangle$, where $d_n = \deg(\Phi_n(x))$ and $k \in \mathbb{N}$.
\end{Definition}

The irreducibility of cyclotomic polynomials over finite fields is governed by the following criterion.

\begin{Lemma}[\cite{Wu2016}]\label{lem:wu_irred}
The cyclotomic polynomial $\Phi_n(x)$ is irreducible in $\mathbb{F}_q[x]$ if and only if $q$ is a primitive root modulo $n$ and $n \in \{2, 4, r^k, 2r^k\}$, where $r$ is an odd prime and $k \in \mathbb{N}$.
\end{Lemma}

Applying Lemma~\ref{lem:wu_irred} to $n = 7$ yields field conditions under which $\Phi_7(x)$ remains irreducible over $\mathbb{F}_{p^m}$.

\begin{Proposition}\label{prop:phi7_irred}
The $7$-th cyclotomic polynomial $\Phi_7(x) = x^6 + x^5 + x^4 + x^3 + x^2 + x + 1$ is irreducible in $\mathbb{F}_{p^m}[x]$ if and only if $p \equiv 3 \pmod 7$ or $p \equiv 5 \pmod 7$, with $m$ being an odd integer such that $\gcd(m, 6) = 1$.
\end{Proposition}

\begin{proof}
By Lemma~\ref{lem:wu_irred}, $\Phi_7(x)$ is irreducible over $\mathbb{F}_{p^m}$ if and only if $p^m$ is a primitive root modulo $7$. The group of units $\mathbb{Z}_7^\times$ is cyclic of order $6$, generated by $3$ and $5 \pmod 7$. Thus, $p^m$ is a primitive root modulo $7$ if and only if $p^m \equiv 3 \pmod 7$ or $p^m \equiv 5 \pmod 7$.

We analyze $p^m \pmod 7$ across all possible congruences of $p \pmod 7$:\\
 \textbf{Case 1:} If $p \equiv 0, 1, 2, 4, \text{ or } 6 \pmod 7$, then $p^m \pmod 7$ generates a proper subgroup of $\mathbb{Z}_7^\times$ of order dividing $3$ for any $m \in \mathbb{N}$. Thus, in this case $p^m$ cannot be a primitive root modulo $7$.\\
\textbf{Case 2:} If $p \equiv 3 \pmod 7$, then $p^m \equiv 3^m \pmod 7$. If $m$ is even ($m = 2k$), then $p^m \equiv (3^2)^k \equiv 2^k \pmod 7 \in \{1, 2, 4\}$, which does not generate $\mathbb{Z}_7^\times$. If $m$ is odd and satisfies $\gcd(m, 6) = 1$, then $\gcd(m, |\mathbb{Z}_7^\times|) = 1$, implying $p^m \equiv 3^m \pmod 7 \in \{3, 5\}$ retains multiplicative order $6$. Hence, $p^m$ is a primitive root modulo $7$.\\
   \textbf{Case 3:} Similarly, if $p \equiv 5 \pmod 7$, then $p^m \equiv 5^m \pmod 7$ is a primitive root modulo $7$ if and only if $m$ is odd with $\gcd(m, 6) = 1$.

Combining these cases completes the proof.
\end{proof}

\subsection{Repeated-Root Cyclic Codes of Length $p^s$ over $R_2$}

The classification and cardinalities of repeated-root cyclic codes of length $p^s$ over $R_2$ established by Dinh \cite{Dinh2010} serve as a key component for our direct sum decomposition.

\begin{Theorem}[\cite{Dinh2010}]\label{thm:dinh_ps_classification}
Let $g'(x) = x - 1$. The cyclic codes of length $p^s$ over $R_2 = \mathbb{F}_{p^m} + u\mathbb{F}_{p^m}$ are partitioned into four disjoint structural types:
\begin{enumerate}
    \item [(1)]\textbf{Type 1:} $C'_1 = \langle 0 \rangle$ or $C'_1 = \langle 1 \rangle$.
    \item [(2)]\textbf{Type 2:} $C'_2(\beta) = \langle u g'(x)^\beta \rangle$, where $0 \le \beta \le p^s - 1$.
    \item [(3)] \textbf{Type 3:} $C'_3(\beta, i, a(x)) = \langle g'(x)^\beta + u g'(x)^i a(x) \rangle$, where $p^s > \beta > i \ge 0$, and $a(x)$ is $0$ or a unit in $R_2[x]/\langle g'(x)^{p^s} \rangle$ of the form $a(x) = \sum_{j=0}^{p^s - i - 1} a_j g'(x)^j$ with $a_j \in \mathbb{F}_{p^m}$ and $a_0 \neq 0$.
    \item [(4)] \textbf{Type 4:} $C'_4(\beta, i, a(x), \gamma) = \langle g'(x)^\beta + u g'(x)^i a(x), u g'(x)^\gamma \rangle$, where $p^s > \beta \ge W > \gamma > i \ge 0$, and $W = \min \{ k \in \mathbb{N}_0 \mid u g'(x)^k \in C'_3(\beta, i, a(x)) \}$.
\end{enumerate}
\end{Theorem}

\begin{Proposition}[\cite{Dinh2010}]\label{prop:W_val}
The index parameter $W = \min \{ k \in \mathbb{N} \mid u g'(x)^k \in \langle g'(x)^\beta + u g'(x)^i a(x) \rangle \}$ is given by
\begin{equation*}
W = \begin{cases} \beta, & \text{if } a(x) = 0 \\ \min\{\beta, p^s - \beta + i\}, & \text{if } a(x) \neq 0 \end{cases}
\end{equation*}
\end{Proposition}

\begin{Theorem}[\cite{Dinh2010}]\label{thm:size_ps_cyclic}
Let $C_1$ be a cyclic code of length $p^s$ over $R_2$. The cardinality $n_{C_1} = |C_1|$ is given by:
\begin{itemize}
    \item [(1)] If $C_1 = \langle 0 \rangle$, then $n_{C_1} = 1$.
    \item [(2)]If $C_1 = \langle 1 \rangle$, then $n_{C_1} = p^{2m p^s}$.
    \item [(3)]If $C_1 = \langle u g'(x)^\beta \rangle$ ($0 \le \beta \le p^s - 1$), then $n_{C_1} = p^{m(p^s - \beta)}$.
    \item [(4)]If $C_1 = \langle g'(x)^\beta \rangle$ ($1 \le \beta \le p^s - 1$), then $n_{C_1} = p^{2m(p^s - \beta)}$.
    \item [(5)]If $C_1 = \langle g'(x)^\beta + u g'(x)^i a(x) \rangle$ with $a(x)$ invertible, then
    \begin{equation*}
    n_{C_1} = \begin{cases} p^{2m(p^s - \beta)}, & \text{if } 1 \le \beta \le \lfloor \frac{p^s + i}{2} \rfloor \\[1ex] p^{m(2p^s - 2\beta + i)}, & \text{if } \lfloor \frac{p^s + i}{2} \rfloor < \beta \le p^s - 1 \end{cases}
    \end{equation*}
    \item [(6)]If $C_1 = \langle g'(x)^\beta + u g'(x)^i a(x), u g'(x)^\gamma \rangle$, then $n_{C_1} = p^{m(2p^s - \beta - \gamma)}$.
\end{itemize}
\end{Theorem}

\subsection{ Gray Map}

To project linear codes over $R_2$ down to linear codes over the base field $\mathbb{F}_{p^m}$, we utilize the canonical $\mathbb{F}_{p^m}$-linear \cite{Dinh2010}Gray map $\Phi: R_2^n \to \mathbb{F}_{p^m}^{2n}$ defined componentwise by
\begin{equation*}
\Phi(\mathbf{a} + u\mathbf{b}) = (\mathbf{b}, \, \mathbf{a} + \mathbf{b}), \;\textrm{for all} \;\mathbf{a}, \mathbf{b} \in \mathbb{F}_{p^m}^n
\end{equation*}

\begin{Lemma}[\cite{Dinh2010, Chen2014}]\label{lem:gray_map}
The Gray map $\Phi$ is an $\mathbb{F}_{p^m}$-linear bijection possessing the following properties:
\begin{enumerate}
    \item [(1)]If $C$ is a linear code of length $n$ over $R_2$ of cardinality $|C| = p^{m k_0}$, then its Gray image $\Phi(C)$ is a linear $[2n, k_0, d_H]_{p^m}$ code over $\mathbb{F}_{p^m}$, where $d_H$ denotes the minimum Hamming distance of $\Phi(C)$.
    \item [(2)] For any linear code $C \subseteq R_2^n$, $\Phi(C^\perp) = \Phi(C)^\perp$ under the Euclidean inner product. In particular, $C^\perp \subseteq C$ over $R_2$ if and only if $\Phi(C)^\perp \subseteq \Phi(C)$ over $\mathbb{F}_{p^m}$.
\end{enumerate}
\end{Lemma}

\subsection{Quantum Error-Correcting Codes and LCD Codes}

Quantum error-correcting codes protect quantum information against environmental decoherence and operational noise. A $q$-ary quantum stabilizer code of length $N$ with $K$ logical qubits and minimum distance $d_q$ is denoted by $[N, K, d_q]_q$. The Calderbank-Shor-Steane (CSS) construction provides a bridge between classical dual-containing linear codes and quantum stabilizer codes.

\begin{Theorem}[CSS Construction \cite{Calderbank1996,Steane1996, Ling2004}]\label{thm:css_construction}
Let $C$ be a linear code of length $n$ over $R_2$ such that $C^\perp \subseteq C$. If $|\Phi(C)| = p^{m k_0}$, then its Gray image $\Phi(C)$ yields an $[2n, \, 2k_0 - 2n, \, d_q]_{p^m}$ quantum stabilizer code over $\mathbb{F}_{p^m}$, where $d_q = \operatorname{wt}_H(\Phi(C) \setminus \Phi(C)^\perp) \ge d_H(\Phi(C))$.
\end{Theorem}

\cite{Massey1992} Linear codes with complementary duals (LCD codes) over  $\mathbb{F}_{p^m}$ are linear codes whose intersection with their Euclidean dual code is trivial.

\begin{Definition}\cite{Xh2015}
A linear code $C$ over a finite chain ring is called a \textbf{linear code with complementary dual (LCD code)} if $C \cap C^\perp = \{\mathbf{0}\}$.
\end{Definition}

\begin{Lemma}\cite{Islam2021}\label{lem:lcd_gray_map}
Let $C$ be a linear code over $R_2$. Then $C$ is an LCD code over $R_2$ if and only if its Gray image $\Phi(C)$ is an LCD code over $\mathbb{F}_{p^m}$.
\end{Lemma}

\begin{proof}
By Lemma~\ref{lem:gray_map}, $\Phi$ is an $\mathbb{F}_{p^m}$-linear bijection and preserves Euclidean duality, i.e., $\Phi(C^\perp) = \Phi(C)^\perp$. Consequently,
$\Phi(C \cap C^\perp) = \Phi(C) \cap \Phi(C^\perp) = \Phi(C) \cap \Phi(C)^\perp$. Since $\Phi$ is bijective, $C \cap C^\perp = \{\mathbf{0}\}$ if and only if $\Phi(C) \cap \Phi(C)^\perp = \{\mathbf{0}\}$, establishing the equivalence.
\end{proof}

\section{\textbf{Cyclic Codes over $\mathbb{F}_{p^m}+u\mathbb{F}_{p^m}$}}
A cyclic code $C$ of length $7p^s$ over $R_2$ is identified as an ideal of the quotient ring $R_2[x]/\langle x^{7p^s} - 1 \rangle$. When $p \equiv 3 \pmod 7$ or $p \equiv 5 \pmod 7$ with $m$ being odd and $\gcd(m, 6) = 1$, the algebraic structure of $R_2[x]/\langle x^{7p^s} - 1 \rangle$ admits a convenient direct sum decomposition.

\begin{Theorem} \label{prop:direct_sum_decomp}
Let $p \equiv 3 \pmod 7$ or $p \equiv 5 \pmod 7$, and let $m$ be an odd integer such that $\gcd(m, 6) = 1$. Then any cyclic code $C$ of length $7p^s$ over $R_2$ decomposes uniquely as
\begin{equation*}
C = C_1 \oplus C_2
\end{equation*}
where $C_1$ is a cyclic code of length $p^s$ over $R_2$ and $C_2$ is a $7$-cyclotomic code of length $6p^s$ over $R_2$.
\end{Theorem}

\begin{proof}
In $\mathbb{F}_{p^m}[x]$, we have the factorization  $x^{7p^s} - 1 = (x^7 - 1)^{p^s} = (x - 1)^{p^s} (\Phi_7(x))^{p^s}$,  where $\Phi_7(x) = x^6 + x^5 + x^4 + x^3 + x^2 + x + 1$ is the $7$-th cyclotomic polynomial. Under the given conditions on $p$ and $m$, Proposition~\ref{prop:phi7_irred} guarantees that $\Phi_7(x)$ is irreducible over $\mathbb{F}_{p^m}$. Because $\gcd(x - 1, \Phi_7(x)) = 1$ in $\mathbb{F}_{p^m}[x]$, by Hensel's Lemma it implies that $\langle (x - 1)^{p^s} \rangle$ and $\langle (\Phi_7(x))^{p^s} \rangle$ are coprime ideals in $R_2[x]$. Consequently, by the Chinese Remainder Theorem, we obtain that

$R_2[x]/\langle x^{7p^s} - 1 \rangle \cong R_2[x]/\langle (x - 1)^{p^s} \rangle \oplus R_2[x]/\langle (\Phi_7(x))^{p^s} \rangle$.
Since every ideal of a direct sum of rings decomposes as a direct sum of ideals from each component ring, any cyclic code $C$ over $R_2$ of length $7p^s$ splits uniquely as $C = C_1 \oplus C_2$, where $C_1 \subseteq R_2[x]/\langle (x - 1)^{p^s} \rangle$ and $C_2 \subseteq R_2[x]/\langle (\Phi_7(x))^{p^s} \rangle$.
\end{proof}

\subsection{\textbf{Classification of Cyclotomic Codes of Length $7p^s$ over $\mathbb{F}_{p^m}+u\mathbb{F}_{p^m}$}}

By Proposition~\ref{prop:direct_sum_decomp}, the complete classification of cyclic codes of length $7p^s$ over $R_2$ reduces to classifying cyclic codes of length $p^s$ over $R_2$ and $7$-cyclotomic codes of length $6p^s$ over $R_2$. 

Let $g(x) = \Phi_7(x) = x^6 + x^5 + x^4 + x^3 + x^2 + x + 1$. We first establish the structural classification of $7$-cyclotomic codes of length $6p^s$ over $R_2$.

\begin{Theorem}\label{thm:classification_7cyclotomic}
Let $p \equiv 3 \pmod 7$ or $p \equiv 5 \pmod 7$, and let $m$ be an odd integer with $\gcd(m, 6) = 1$. Let $g(x) = x^6 + x^5 + x^4 + x^3 + x^2 + x + 1$. The $7$-cyclotomic codes of length $6p^s$ over $R_2$ (i.e., ideals of $R_2[x]/\langle g(x)^{p^s} \rangle$) are partitioned into four mutually disjoint types:
\begin{enumerate}
    \item [(1)]\textbf{Type 1:} $C = \langle 0 \rangle$ or $C = \langle 1 \rangle$.
    \item [(2)]\textbf{Type 2:} $C_2(\beta) = \langle u g(x)^\beta \rangle$, where $0 \le \beta \le p^s - 1$.
    \item [(3)]\textbf{Type 3:} $C_3(\beta, i, a(x)) = \langle g(x)^\beta + u g(x)^i a(x) \rangle$, where $p^s > \beta > i \ge 0$, and $a(x)$ is either $0$ or a unit in $R_2[x]/\langle g(x)^{p^s} \rangle$ represented as
    \begin{equation*}
    a(x) = \sum_{j=0}^{p^s - i - 1} a_j (x) g(x)^j
    \end{equation*}
    with $a_j(x) \in \mathbb{F}_{p^m}[x]$, $\deg(a_j(x)) \le 1$, and $a_0(x) \neq 0$.
    \item [(4)]\textbf{Type 4:} $C_4(\beta, i, a(x), \gamma) = \langle g(x)^\beta + u g(x)^i a(x), u g(x)^\gamma \rangle$, where $p^s > \beta \ge V > \gamma > i \ge 0$, $a(x)$ is $0$ or a unit as defined in Type 3, and $V$ is the minimum integer such that $u g(x)^V \in C_3(\beta, i, a(x))$.
\end{enumerate}
\end{Theorem}

\begin{proof}
Under the hypothesis that $p \equiv 3 \pmod 7$ or $p \equiv 5 \pmod 7$ with $\gcd(m, 6) = 1$, the polynomial $g(x) = \Phi_7(x) = x^6 + x^5 + x^4 + x^3 + x^2 + x + 1$ is irreducible over $\mathbb{F}_{p^m}$ of degree $6$. Thus, the residue field $K = \mathbb{F}_{p^m}[x]/\langle g(x) \rangle \cong \mathbb{F}_{p^{6m}}$ is a degree-$6$ extension over $\mathbb{F}_{p^m}$. Letting $\pi = g(x) \pmod{\langle g(x)^{p^s} \rangle}$, the ambient ring $\mathcal{A} = R_2[x]/\langle g(x)^{p^s} \rangle$ is canonically isomorphic to the local bivariate quotient algebra:
\begin{equation*}
\mathcal{A} \cong \frac{K[\pi, u]}{\langle \pi^{p^s}, u^2 \rangle},
\end{equation*}
where $u^2 = 0$, $\pi^{p^s} = 0$, and scalars from $K$ commute with both $u$ and $\pi$. The subring $\mathcal{S} = K[\pi]/\langle \pi^{p^s} \rangle$ is a finite principal ideal chain ring whose ideals form the unique sequence $\langle 0 \rangle = \langle \pi^{p^s} \rangle \subsetneq \langle \pi^{p^s-1} \rangle \subsetneq \dots \subsetneq \langle \pi \rangle \subsetneq \langle 1 \rangle = \mathcal{S}$. Every element of $\mathcal{A}$ has a unique representation as $f(\pi) + u h(\pi)$ with $f(\pi), h(\pi) \in \mathcal{S}$. Consider the canonical projection $\phi: \mathcal{A} \to \mathcal{S}$ given by $\phi(f(\pi) + u h(\pi)) = f(\pi)$. For any ideal $C \subseteq \mathcal{A}$, the residue submodule $\operatorname{Res}(C) = \phi(C) = \langle \pi^\beta \rangle$ and the torsion submodule $\operatorname{Tor}(C) = \{ h(\pi) \in \mathcal{S} \mid u h(\pi) \in C \} = \langle \pi^\gamma \rangle$ are ideals of $\mathcal{S}$, defined by unique integers $\beta, \gamma \in \{0, 1, \dots, p^s\}$. Because $u C \subseteq C$, we have $u f(\pi) \in C$ for every $f(\pi) \in \operatorname{Res}(C)$, which yields $\operatorname{Res}(C) \subseteq \operatorname{Tor}(C)$ and forces $\gamma \le \beta$. We now classify all possible ideals into four disjoint types based on the invariant indices $\beta$ and $\gamma$.

\medskip
\noindent\textbf{Type 1:} If $\beta = p^s$ and $\gamma = p^s$, then $\operatorname{Res}(C) = \langle 0 \rangle$ and $\operatorname{Tor}(C) = \langle 0 \rangle$, which gives the trivial zero ideal $C = \langle 0 \rangle$. If $\beta = 0$, then $\operatorname{Res}(C) = \langle 1 \rangle = \mathcal{S}$, meaning $C$ contains an element of the form $1 + u h(\pi)$ for some $h(\pi) \in \mathcal{S}$. Since $(1 + u h(\pi))(1 - u h(\pi)) = 1 - u^2 h(\pi)^2 = 1$, the element $1 + u h(\pi)$ is an invertible unit in $\mathcal{A}$. An ideal containing a unit equals the entire ambient algebra, so $C = \langle 1 \rangle$. These constitute the two trivial ideals of \textbf{Type 1}.

\medskip
\noindent\textbf{Type 2:} If $\beta = p^s$ and $0 \le \gamma \le p^s - 1$, then $\operatorname{Res}(C) = \langle 0 \rangle$, so every element of $C$ lies in $u\mathcal{S}$. Consequently, $C = u \operatorname{Tor}(C) = u \langle \pi^\gamma \rangle = \langle u \pi^\gamma \rangle$. Replacing $\pi$ with $g(x) \pmod{\langle g(x)^{p^s} \rangle}$ gives the principal ideals $C_2(\gamma) = \langle u g(x)^\gamma \rangle$ of \textbf{Type 2}.

\medskip
\noindent\textbf{Type 3:} Assume $0 < \beta < p^s$. Because $\operatorname{Res}(C) = \langle \pi^\beta \rangle$, there exists an element $g_1 \in C$ of the form $g_1 = \pi^\beta + u \pi^i a(\pi)$, where $i \ge 0$, and $a(\pi) \in \mathcal{S}$ is either zero or an invertible unit with $a(0) \ne 0$. If $a(\pi) \ne 0$ and $i \ge \beta$, then multiplying $g_1$ by $u$ gives $u \pi^\beta \in C$, so $u \pi^i a(\pi) \in C$; subtracting this from $g_1$ yields $\pi^\beta \in C$, which reduces $g_1$ to the case $a(\pi) = 0$. Thus, we assume $\beta > i \ge 0$ whenever $a(\pi) \ne 0$. A torsion element $u h(\pi)$ belongs to $\langle g_1 \rangle$ if and only if $(c_0(\pi) + u c_1(\pi))(\pi^\beta + u \pi^i a(\pi)) = u h(\pi)$ for some $c_0(\pi), c_1(\pi) \in \mathcal{S}$. Expanding this equation shows that $c_0(\pi) \pi^\beta = 0$, forcing $c_0(\pi) = k(\pi) \pi^{p^s - \beta}$, which implies $u h(\pi) = u [k(\pi) \pi^{p^s - \beta + i} a(\pi) + c_1(\pi) \pi^\beta]$. Therefore, the torsion ideal of $\langle g_1 \rangle$ is principal and given by $\operatorname{Tor}(\langle g_1 \rangle) = \langle \pi^V \rangle$, where $V = \beta$ if $a(\pi) = 0$ and $V = \min\{\beta, p^s - \beta + i\}$ if $a(\pi) \ne 0$. When $\gamma = V$, we have $\operatorname{Tor}(C) = \operatorname{Tor}(\langle g_1 \rangle)$. Any arbitrary element $c \in C$ can be written as $c = f_0(\pi) g_1 + u h_0(\pi)$; since $u h_0(\pi) \in \operatorname{Tor}(C) = \operatorname{Tor}(\langle g_1 \rangle) \subseteq \langle g_1 \rangle$, we obtain $c \in \langle g_1 \rangle$, so $C$ is a principal ideal of the form $C = \langle \pi^\beta + u \pi^i a(\pi) \rangle$. Substituting $\pi = g(x) \pmod{\langle g(x)^{p^s} \rangle}$ gives the codes $C_3(\beta, i, a(x))$ of \textbf{Type 3}.

\medskip
\noindent\textbf{Type 4:} Assume $0 < \beta < p^s$ and $\gamma < V$, with $V$ as defined above. Since $\operatorname{Tor}(\langle g_1 \rangle) = \langle \pi^V \rangle \subsetneq \langle \pi^\gamma \rangle = \operatorname{Tor}(C)$, the primary generator $g_1 = \pi^\beta + u \pi^i a(\pi)$ cannot span the entire torsion submodule of $C$, necessitating the minimal torsion generator $u \pi^\gamma$. If $\gamma \le i$, then $\pi^i \in \langle \pi^\gamma \rangle = \operatorname{Tor}(C)$, which implies $u \pi^i a(\pi) \in C$ and reduces $g_1$ to $\pi^\beta$ (forcing $a(\pi) = 0$ and $V = \beta$), contradicting $\gamma < V$; hence, we must have $\gamma > i$. Finally, for any $c \in C$, its projection satisfies $\phi(c) \in \operatorname{Res}(C) = \langle \pi^\beta \rangle$, so there exists $f(\pi) \in \mathcal{S}$ such that $c - f(\pi) g_1 \in \ker(\phi) \cap C = u \operatorname{Tor}(C) = \langle u \pi^\gamma \rangle$, confirming that $c \in \langle g_1, u \pi^\gamma \rangle$. Hence, $C$ is minimally generated by two elements:
\begin{equation*}
C = \langle \pi^\beta + u \pi^i a(\pi), \, u \pi^\gamma \rangle,
\end{equation*}
subject to the strict condition $p^s > \beta \ge V > \gamma > i \ge 0$. Substituting $\pi = g(x) \pmod{\langle g(x)^{p^s} \rangle}$ yields the non-principal codes $C_4(\beta, i, a(x), \gamma)$ of \textbf{Type 4}.

\medskip
Because the invariant pair $(\beta, \gamma)$ and the unit status of $a(\pi)$ uniquely characterize the residue and torsion submodules of $C$, the four types are mutually disjoint and exhaust all ideals of $\mathcal{A}$.
\end{proof}

\begin{Proposition}\label{prop:V_val}
Let $g(x) = x^6 + x^5 + x^4 + x^3 + x^2 + x + 1$. The index parameter $V = \min \{ k \in \mathbb{N} \mid u g(x)^k \in \langle g(x)^\beta + u g(x)^i a(x) \rangle \}$ is given explicitly by
\begin{equation*}
V = \begin{cases} \beta, & \text{if } a(x) = 0 \\ \min\{\beta, p^s - \beta + i\}, & \text{if } a(x) \neq 0 \end{cases}
\end{equation*}
\end{Proposition}

\begin{proof}
If $a(x) = 0$, the generator is $g(x)^\beta$. Multiplying by $u$ gives $u g(x)^\beta \in C$, so $V = \beta$. If $a(x) \neq 0$, $a(x)$ is a unit. Multiplying $g(x)^\beta + u g(x)^i a(x)$ by $u$ yields $u g(x)^\beta \in C$, showing $V \le \beta$. Furthermore, multiplying by $g(x)^{p^s - \beta}$ yields
\begin{equation*}
g(x)^{p^s} + u g(x)^{p^s - \beta + i} a(x) \equiv u g(x)^{p^s - \beta + i} a(x) \pmod{g(x)^{p^s}}
\end{equation*}
Since $a(x)$ is invertible, $u g(x)^{p^s - \beta + i} \in C$, which implies $V \le p^s - \beta + i$. Minimal linear dependence over $\mathbb{F}_{p^m}[x]$ establishes $V = \min\{\beta, p^s - \beta + i\}$.
\end{proof}

For completeness, we recall the classification of cyclic codes of length $p^s$ over $R_2$ established by Dinh \cite{Dinh2010}.

\subsection{\textbf{Enumeration of Codewords via Residue and Torsion Code Techniques}}

To compute the precise size $|C|$ of a code $C$ over $R_2$, we employ the residue and torsion code maps. For a linear code $C$ of length $n$ over $R_2$, its residue code $\operatorname{Res}(C)$ and torsion code $\operatorname{Tor}(C)$ are subcodes over $\mathbb{F}_{p^m}$ defined as:

$$\operatorname{Res}(C) = \{ a \in \mathbb{F}_{p^m}^n \mid \exists\, b \in \mathbb{F}_{p^m}^n \text{ such that } a + ub \in C \}  \;\textrm{and}
\;\operatorname{Tor}(C) = \{ a \in \mathbb{F}_{p^m}^n \mid ua \in C \}.$$

\begin{Theorem}\label{thm:rank_nullity_size}
Let $C$ be a linear code of length $n$ over $R_2$. Then
$$|C| = |\operatorname{Res}(C)| \cdot |\operatorname{Tor}(C)|.$$
\end{Theorem}

\begin{proof}
The proof follows directly from the First Isomorphism Theorem applied to the projection map $\phi: C \to \operatorname{Res}(C)$, whose kernel is isomorphic to $\operatorname{Tor}(C)$. Taking dimensions over $\mathbb{F}_{p^m}$ gives $\dim_{\mathbb{F}_{p^m}}(C) = \dim_{\mathbb{F}_{p^m}}(\operatorname{Res}(C)) + \dim_{\mathbb{F}_{p^m}}(\operatorname{Tor}(C))$, from which the cardinality identity follows immediately.
\end{proof}
\begin{Lemma}\label{lem:res_tor_7cyclotomic}
Let $C$ be a $7$-cyclotomic code of length $6p^s$ over $R_2$. The residue code $\operatorname{Res}(C)$ and torsion code $\operatorname{Tor}(C)$ are ideals of $\mathbb{F}_{p^m}[x]/\langle g(x)^{p^s} \rangle$, determined as follows:
\begin{enumerate}
    \item [(1)]\textbf{Type 1:} If $C = \langle 0 \rangle$, $\operatorname{Res}(C) = \operatorname{Tor}(C) = \langle 0 \rangle$. If $C = \langle 1 \rangle$, $\operatorname{Res}(C) = \operatorname{Tor}(C) = \langle 1 \rangle$.
    \item [(2)]\textbf{Type 2:} If $C = \langle u g(x)^\beta \rangle$ ($0 \le \beta \le p^s - 1$), then $\operatorname{Res}(C) = \langle 0 \rangle$ and $\operatorname{Tor}(C) = \langle g(x)^\beta \rangle$.
    \item [(3)]\textbf{Type 3:} If $C = \langle g(x)^\beta + u g(x)^i a(x) \rangle$ ($p^s > \beta > i \ge 0$), then $\operatorname{Res}(C) = \langle g(x)^\beta \rangle$ and $\operatorname{Tor}(C) = \langle g(x)^V \rangle$, where $V = \beta$ if $a(x) = 0$ and $V = \min\{\beta, p^s - \beta + i\}$ if $a(x) \neq 0$.
    \item [(4)]\textbf{Type 4:} If $C = \langle g(x)^\beta + u g(x)^i a(x), u g(x)^\gamma \rangle$ ($\gamma < V$), then $\operatorname{Res}(C) = \langle g(x)^\beta \rangle$ and $\operatorname{Tor}(C) = \langle g(x)^\gamma \rangle$.
\end{enumerate}
\end{Lemma}

For any ideal $\langle g(x)^k \rangle \subseteq \mathbb{F}_{p^m}[x]/\langle g(x)^{p^s} \rangle$ (where $\deg(g(x)) = 6$), the size of the corresponding cyclic code over $\mathbb{F}_{p^m}$ is $|\langle g(x)^k \rangle| = p^{6m(p^s - k)}$. Combining Theorem~\ref{thm:rank_nullity_size} and Lemma~\ref{lem:res_tor_7cyclotomic} yields the exact codeword counts.

The proof of the following theorem proceeds analogously to that of Lemma~2.3 in \cite{Dinh2012}. So its proof is omitted.
\begin{Theorem}\label{thm:size_7cyclotomic}
Let $C$ be a $7$-cyclotomic code of length $6p^s$ over $R_2$. The cardinality $n_C = |C|$ is given by:
\begin{itemize}
    \item [(1)]If $C = \langle 0 \rangle$, then $n_C = 1$.
    \item [(2)]If $C = \langle 1 \rangle$, then $n_C = p^{12m p^s}$.
    \item [(3)]If $C = \langle u g(x)^\beta \rangle$ ($0 \le \beta \le p^s - 1$), then $n_C = p^{6m(p^s - \beta)}$.
    \item [(4)] If $C = \langle g(x)^\beta \rangle$ ($1 \le \beta \le p^s - 1$), then $n_C = p^{12m(p^s - \beta)}$.
    \item [(5)]If $C = \langle g(x)^\beta + u g(x)^i a(x) \rangle$ with $a(x)$ invertible, then
    \begin{equation*}
    n_C = \begin{cases} p^{12m(p^s - \beta)}, & \text{if } 1 \le \beta \le \lfloor \frac{p^s + i}{2} \rfloor \\ p^{6m(2p^s - 2\beta + i)}, & \text{if } \lfloor \frac{p^s + i}{2} \rfloor < \beta \le p^s - 1 \end{cases}
    \end{equation*}
    \item[(6)] If $C = \langle g(x)^\beta + u g(x)^i a(x), u g(x)^\gamma \rangle$, then $n_C = p^{6m(2p^s - \beta - \gamma)}$.
\end{itemize}
\end{Theorem}

By Proposition~\ref{prop:direct_sum_decomp}, the cardinality of any cyclic code $C = C_1 \oplus C_2$ of length $7p^s$ over $R_2$ is $n_C = n_{C_1} n_{C_2}$, where $n_{C_1}$ and $n_{C_2}$ are given by Theorems~\ref{thm:size_7cyclotomic} and \ref{thm:size_ps_cyclic}. For matched component parameters, $n_C$ is evaluated as follows.

\begin{Theorem}\label{thm:size_total}
Let $C = C_1 \oplus C_2$ be a cyclic code of length $7p^s$ over $R_2 = \mathbb{F}_{p^m} + u\mathbb{F}_{p^m}$. Then $|C| = |C_1| \cdot |C_2|$. In particular, for matched component code types with identical parameters, the cardinality $n_C = |C|$ is given by:
\begin{enumerate}
    \item [(1)] If $C = \langle 0 \rangle$, then $n_C = 1$.
    \item [(2)]If $C = \langle 1 \rangle$, then $n_C = p^{14mp^s}$.
    \item [(1)]If $C = C'_2(\beta) \oplus C_2(\beta)$ for $0 \le \beta \le p^s - 1$, then $n_C = p^{7m(p^s - \beta)}$.
    \item [(3)]If $C = C'_3(\beta, 0, 0) \oplus C_3(\beta, 0, 0) = \langle g'(x)^\beta \rangle \oplus \langle g(x)^\beta\rangle$ for $1 \le \beta \le p^s - 1$, then $n_C = p^{14m(p^s - \beta)}$.
    \item [(4)]If $C = C'_3(\beta, i, a(x)) \oplus C_3(\beta, i, a(x))$ with $p^s > \beta > i \ge 0$ and $a(x) \in R_2[x]^\times$, then
    \begin{equation*}
    n_C = \begin{cases} 
    p^{14m(p^s - \beta)}, & \text{if } 1 \le \beta \le \left\lfloor \frac{p^s + i}{2} \right\rfloor \\[1ex]
    p^{7m(p^s - i)}, & \text{if } \left\lfloor \frac{p^s + i}{2} \right\rfloor < \beta \le p^s - 1
    \end{cases}
    \end{equation*}
    \item [(5)]If $C = C'_4(\beta, i, a(x), \gamma) \oplus C_4(\beta, i, a(x), \gamma)$ with $\gamma < \min\{\beta, p^s - \beta + i\}$, then $n_C = p^{7m(2p^s - \beta - \gamma)}$.
\end{enumerate}
\end{Theorem}

\begin{proof}
By Proposition~\ref{prop:direct_sum_decomp}, any cyclic code $C$ of length $7p^s$ over $R_2$ satisfies $C \cong C_1 \oplus C_2$ as $R_2$-modules, where $C_1$ is a cyclic code of length $p^s$ over $R_2$ and $C_2$ is a $7$-cyclotomic code of length $6p^s$ over $R_2$. Consequently, the total cardinality satisfies $|C| = |C_1| \cdot |C_2| = n_{C_1} \cdot n_{C_2}$. 

\textbf{(1):} $C = \langle 0 \rangle$. \\
Here, $C_1 = \langle 0 \rangle \subseteq R_2[x]/\langle g'(x)^{p^s} \rangle$ and $C_2 = \langle 0 \rangle \subseteq R_2[x]/\langle g(x)^{p^s} \rangle$. By Theorem~\ref{thm:size_7cyclotomic} and \ref{thm:size_ps_cyclic}, $n_{C_1} = 1$ and $n_{C_2} = 1$. Thus,
$n_C = n_{C_1} \cdot n_{C_2} = 1 \cdot 1 = 1$.
\vspace{1em}

\textbf{( 2):} $C = \langle 1 \rangle$. \\
Here, $C_1 = \langle 1 \rangle$ and $C_2 = \langle 1 \rangle$. By Theorem~\ref{thm:size_7cyclotomic} and \ref{thm:size_ps_cyclic}, $n_{C_1} = p^{2mp^s}$ and $n_{C_2} = p^{12mp^s}$. Therefore, $n_C = p^{2mp^s} \cdot p^{12mp^s} = p^{(2m + 12m)p^s} = p^{14mp^s}$.
\vspace{1em}

\textbf{(3):} $C = C'_2(\beta) \oplus C_2(\beta)$ for $0 \le \beta \le p^s - 1$. \\
The components are $C_1 = \langle u g'(x)^\beta \rangle$ and $C_2 = \langle u g(x)^\beta \rangle$. By Lemma~\ref{lem:res_tor_7cyclotomic} and its $p^s$-analog, $\operatorname{Res}(C_1) = \operatorname{Res}(C_2) = \langle 0 \rangle$, while $\operatorname{Tor}(C_1) = \langle g'(x)^\beta \rangle \subseteq \mathbb{F}_{p^m}[x]/\langle g'(x)^{p^s} \rangle$ and $\operatorname{Tor}(C_2) = \langle g(x)^\beta \rangle \subseteq \mathbb{F}_{p^m}[x]/\langle g(x)^{p^s} \rangle$. Since $\deg(g'(x)) = 1$ and $\deg(g(x)) = 6$, we have $n_{C_1} = |\operatorname{Tor}(C_1)| = p^{m(p^s - \beta)}$ and $n_{C_2} = |\operatorname{Tor}(C_2)| = p^{6m(p^s - \beta)}$. Hence,
$n_C = p^{m(p^s - \beta)} \cdot p^{6m(p^s - \beta)} = p^{7m(p^s - \beta)}$.

\textbf{( 4):} $C = C'_3(\beta, 0, 0) \oplus C_3(\beta, 0, 0) = \langle g'(x)^\beta \rangle \oplus \langle g(x)^\beta \rangle$ for $1 \le \beta \le p^s - 1$. \\
For $a(x) = 0$, Lemma~\ref{lem:res_tor_7cyclotomic} gives $\operatorname{Res}(C_1) = \langle g'(x)^\beta \rangle$, $\operatorname{Tor}(C_1) = \langle g'(x)^\beta \rangle$, $\operatorname{Res}(C_2) = \langle g(x)^\beta \rangle$, and $\operatorname{Tor}(C_2) = \langle g(x)^\beta \rangle$. Application of Theorem~\ref{thm:rank_nullity_size} yields $n_{C_1} = p^{m(p^s - \beta)} \cdot p^{m(p^s - \beta)} = p^{2m(p^s - \beta)}$ and $n_{C_2} = p^{6m(p^s - \beta)} \cdot p^{6m(p^s - \beta)} = p^{12m(p^s - \beta)}$. Thus,
$n_C = p^{2m(p^s - \beta)} \cdot p^{12m(p^s - \beta)} = p^{14m(p^s - \beta)}$

\textbf{(5):} $C = C'_3(\beta, i, a(x)) \oplus C_3(\beta, i, a(x))$ with $p^s > \beta > i \ge 0$ and $a(x) \in R_2[x]^\times$. \\
By Lemma~\ref{lem:res_tor_7cyclotomic}, $\operatorname{Res}(C_1) = \langle g'(x)^\beta \rangle$, $\operatorname{Tor}(C_1) = \langle g'(x)^V \rangle$, $\operatorname{Res}(C_2) = \langle g(x)^\beta \rangle$, and $\operatorname{Tor}(C_2) = \langle g(x)^V \rangle$, where $V = \min\{\beta, p^s - \beta + i\}$. By Theorem~\ref{thm:rank_nullity_size}, $n_{C_1} = p^{m(2p^s - \beta - V)}$ and $n_{C_2} = p^{6m(2p^s - \beta - V)}$, giving $n_C = p^{7m(2p^s - \beta - V)}$. We evaluate $V$ under two subcases:
\begin{itemize}
    \item If $1 \le \beta \le \left\lfloor \frac{p^s + i}{2} \right\rfloor$, then $\beta \le p^s - \beta + i$, so $V = \beta$. Substituting $V = \beta$ yields
    $n_C = p^{7m(2p^s - 2\beta)} = p^{14m(p^s - \beta)}$
    \item If $\left\lfloor \frac{p^s + i}{2} \right\rfloor < \beta \le p^s - 1$, then $p^s - \beta + i < \beta$, so $V = p^s - \beta + i$. Substituting $V$ yields
    $n_C = p^{7m(2p^s - \beta - (p^s - \beta + i))} = p^{7m(p^s - i)}$
\end{itemize}

\textbf{( 6:)} $C = C'_4(\beta, i, a(x), \gamma) \oplus C_4(\beta, i, a(x), \gamma)$ with $\gamma < V = \min\{\beta, p^s - \beta + i\}$.

By Lemma~\ref{lem:res_tor_7cyclotomic}, the residue and torsion codes are $\operatorname{Res}(C_1) = \langle g'(x)^\beta \rangle$, $\operatorname{Tor}(C_1) = \langle g'(x)^\gamma \rangle$, $\operatorname{Res}(C_2) = \langle g(x)^\beta \rangle$, and $\operatorname{Tor}(C_2) = \langle g(x)^\gamma \rangle$. Applying Theorem~\ref{thm:rank_nullity_size} yields $n_{C_1} = p^{m(p^s - \beta)} \cdot p^{m(p^s - \gamma)} = p^{m(2p^s - \beta - \gamma)}$ and $n_{C_2} = p^{6m(p^s - \beta)} \cdot p^{6m(p^s - \gamma)} = p^{6m(2p^s - \beta - \gamma)}$. Therefore,
$n_C = p^{m(2p^s - \beta - \gamma)} \cdot p^{6m(2p^s - \beta - \gamma)} = p^{7m(2p^s - \beta - \gamma)}$
\end{proof}

\subsection{\textbf{Dual Codes of $7$-Cyclotomic and Cyclic Codes}}

%Let $C$ be a linear code of length $n$ over $R_2 = \mathbb{F}_{p^m} + u\mathbb{F}_{p^m}$. The Euclidean dual code $C^\perp$ is defined by\begin{equation}C^\perp = \{ \mathbf{y} \in R_2^n \mid \mathbf{x} \cdot \mathbf{y} = 0, \, \forall \mathbf{x} \in C \}\end{equation}Under the canonical polynomial identification, $C^\perp$ corresponds to the annihilator ideal of $C$ involving reciprocal polynomials. For a polynomial $h(x) = \sum_{j=0}^d h_j x^j$ of degree $d$, its reciprocal polynomial is $h_R(x) = x^d h(1/x) = \sum_{j=0}^d h_{d-j} x^j$. 

The structural behavior of dual codes over $R_2[x]/\langle f(x)^{p^s} \rangle$ depends intrinsically on whether $f(x)$ is self-reciprocal. For the two primary factors of $x^{7p^s} - 1$:
\begin{enumerate}
    \item The $7$-th cyclotomic polynomial $g(x) = \Phi_7(x) = x^6 + x^5 + x^4 + x^3 + x^2 + x + 1$ is self-reciprocal, i.e., $g_R(x) = g(x)$.
    \item The linear factor $g'(x) = x - 1$ satisfies $g'_R(x) = 1 - x = -g'(x)$, introducing alternating sign factors $(-1)^k$ into the annihilator generators.
\end{enumerate}

Since $R_2[x]/\langle f(x)^{p^s} \rangle$ is a finite self-dual ring algebra, any ideal $C$ satisfies the fundamental cardinality identity $|C| \cdot |C^\perp| = |R_2[x]/\langle f(x)^{p^s} \rangle|$. We utilize this relation alongside annihilator containment to derive explicit dual generators.

\subsubsection{\textbf{Dual Codes of $7$-Cyclotomic Codes of Length $6p^s$}}

We first examine $7$-cyclotomic codes over $R_2[x]/\langle g(x)^{p^s} \rangle$. For Type 1 (trivial codes) and Type 2 (pure torsion codes), the dual structures follow directly from the orthogonal properties of the nilpotent element $u$.

\begin{Proposition}\label{prop:dual_type1_2}
Let $g(x) = x^6 + x^5 + x^4 + x^3 + x^2 + x + 1$. The duals of Type 1 and Type 2 $7$-cyclotomic codes over $R_2$ are given as follows:
\begin{enumerate}
    \item [(1)]\textbf{Type 1:} $\langle 0 \rangle^\perp = \langle 1 \rangle$ and $\langle 1 \rangle^\perp = \langle 0 \rangle$.
    \item [(2)]\textbf{Type 2:} For $0 \le \beta \le p^s - 1, \;(C_2(\beta))^\perp = (\langle u g(x)^\beta \rangle)^\perp = \langle g(x)^{p^s - \beta}, u \rangle\\ = C_4(p^s - \beta, 0, 0, 0)$
\end{enumerate}
\end{Proposition}

\begin{proof}
 \textbf{(1):} This follows  immediately  from the definition of dual space. 
 
\textbf{ (2):} Let $f(x) = a(x) + u b(x) \in (C_2(\beta))^\perp$. Multiplying $f(x)$ by the generator $u g(x)^\beta$ yields $u g(x)^\beta a(x) \equiv 0 \pmod{g(x)^{p^s}}$. Since $g(x)$ is irreducible over $\mathbb{F}_{p^m}$, $g(x)^{p^s - \beta}$ divides $a(x)$, implying $a(x) \in \langle g(x)^{p^s - \beta} \rangle$. Thus $f(x) \in \langle g(x)^{p^s - \beta}, u \rangle$, so $(C_2(\beta))^\perp \subseteq \langle g(x)^{p^s - \beta}, u \rangle$.

Conversely, for any $h(x) = k_1(x) g(x)^{p^s - \beta} + u k_2(x) \in \langle g(x)^{p^s - \beta}, u \rangle$, direct calculation shows $u g(x)^\beta h(x) \equiv 0 \pmod{g(x)^{p^s}}$. Hence $\langle g(x)^{p^s - \beta}, u \rangle \subseteq (C_2(\beta))^\perp$, completing the proof.
\end{proof}

For single-generator codes of Type 3 containing a non-zero unit part $u g(x)^i a(x)$, the dual generator transitions between a single-generator and a two-generator form depending on whether $2\beta - i$ exceeds the ambient nilpotency bound $p^s$.

\begin{Proposition}\label{prop:dual_type3}
Let $C_3(\beta, i, a(x)) = \langle g(x)^\beta + u g(x)^i a(x) \rangle$ be a Type 3 $7$-cyclotomic code over $R_2$. Its dual code $C_3^\perp$ is determined as follows:
\begin{equation*}
C_3^\perp = \begin{cases} 
C_3(p^s - \beta, 0, 0), & \text{if } a(x) = 0 \\[1ex] 
C_3(p^s - \beta, p^s + i - 2\beta + v, -G(x)), & \text{if } a(x) \neq 0 \text{ and } p^s \ge 2\beta - i \\[1ex] 
C_3(\beta - i, v, -G(x)), & \text{if } a(x) \neq 0 \text{ and } p^s \le 2\beta - i 
\end{cases}
\end{equation*}
where $t = \max\{ k \in \mathbb{N} \mid g(x)^k \text{ divides } x^{6(\beta - i)} a(1/x) \}$ and $x^{6(\beta - i)} a(1/x) = g(x)^tG(x)$ for some unit $G(x) \in (R_2[x]/\langle g(x)^{p^s} \rangle)^\times$.
\end{Proposition}

\begin{proof}
Let $\mathcal{A} = R_2[x]/\langle g(x)^{p^s} \rangle \cong K[\pi, u]/\langle \pi^{p^s}, u^2 \rangle$, where $K = \mathbb{F}_{p^{6m}}$ and $\pi = g(x) \pmod{g(x)^{p^s}}$. Since $g_R(x) = g(x)$, $C_3^\perp = \operatorname{Ann}(C_3)$ and $|C_3| \cdot |C_3^\perp| = p^{12mp^s}$.

\textbf{Case 1:} $a(x) = 0$. \\
Taking $h(\pi) = \pi^{p^s - \beta}$ gives $h(\pi) \pi^\beta = 0$. Since $|C_3| = p^{12m(p^s - \beta)}$ and $|\langle h(\pi) \rangle| = p^{12m\beta}$, size matching confirms $C_3^\perp = C_3(p^s - \beta, 0, 0)$.

\textbf{Case 2:} $a(x) \neq 0$ and $p^s \ge 2\beta - i$. \\
Let $a_R(\pi) = \pi^{6(\beta - i)} a(1/\pi) = \pi^t G(\pi)$. Define $h(\pi) = \pi^{p^s - \beta} - u \pi^{p^s + i - 2\beta + t} G(\pi)$. Then
\begin{align*}
&h(\pi) (\pi^\beta + u \pi^i a(\pi))\\ 
&= \pi^{p^s} + u \left( \pi^{p^s - \beta + i} a(\pi) - \pi^{p^s + i - \beta + t} G(\pi) \right)\\ &\equiv 0 \pmod{\pi^{p^s}}, 
\end{align*}
since $\pi^t G(\pi) \equiv a(\pi) \pmod{\pi^{p^s - i}}$. By Theorem~\ref{thm:size_7cyclotomic}, $|C_3| = p^{12m(p^s - \beta)}$ and $|\langle h(\pi) \rangle| = p^{12m\beta}$, forcing $C_3^\perp = C_3(p^s - \beta, p^s + i - 2\beta + t, -G(x))$.

\textbf{Case 3:} $a(x) \neq 0$ and $p^s \le 2\beta - i$.

Define $h(\pi) = \pi^{\beta - i} - u \pi^t G(\pi)$. Evaluating the product with the generator of $C_3$:
\begin{align*}
&h(\pi)(\pi^\beta + u\pi^i a(\pi))\\ &= \pi^{2\beta - i} + u\pi^\beta a(\pi) - u\pi^{\beta - i + t} G(\pi) \\
&\equiv 0 \pmod{\pi^{p^s}},
\end{align*}
since $2\beta - i \ge p^s$ and $\pi^t G(\pi) \equiv a(\pi) \pmod{\pi^{p^s - i}}$. Thus, $h(\pi) \in \operatorname{Ann}(C_3) = C_3^\perp$.

To determine the cardinality $|\langle h(\pi) \rangle|$, note that $\langle h(\pi) \rangle$ has residue generator $\pi^{\beta - i}$ and torsion generator $\pi^V$, where $V = \min\{\beta - i, \, p^s - \beta + i + t\}$. Since $t = 0$ and $p^s \le 2\beta - i$, we have $V = p^s - \beta + i$. Applying Theorem~\ref{thm:rank_nullity_size}, we obtain that 
$|\langle h(\pi) \rangle| = p^{6m(p^s - (\beta - i))} \cdot p^{6m(p^s - (p^s - \beta + i))} = p^{6m(2\beta - i)}$. By Theorem~\ref{thm:size_7cyclotomic}, $|C_3| = p^{6m(2p^s - 2\beta + i)}$. Consequently,
\begin{equation*}
|C_3| \cdot |\langle h(\pi) \rangle| = p^{6m(2p^s - 2\beta + i)} \cdot p^{6m(2\beta - i)} = p^{12m p^s}
\end{equation*}
Since $|C_3| \cdot |C_3^\perp| = p^{12m p^s}$, size matching confirms $C_3^\perp = \langle h(\pi) \rangle = C_3(\beta - i, t, -G(x))$.

\end{proof}

\begin{Corollary}\label{cor:dual_containing_type3}
Let $C_3(\beta, i, a(x))$ be a Type 3 $7$-cyclotomic code over $R_2$ with $p^s > \beta > i \ge 0$. Then $C_3^\perp \subseteq C_3$ if and only if one of the following conditions holds:
\begin{enumerate}
    \item [(1)]$a(x) = 0$ and $\beta \le \left\lfloor \frac{p^s}{2} \right\rfloor$.
    \item [(2)]$a(x) \neq 0$, $p^s \ge 2\beta - i$, and $p^s + t \ge 2\beta$.
\end{enumerate}
\end{Corollary}

\begin{proof}
For Type 3 codes, $C_3(\beta_2, i_2, a_2(x)) \subseteq C_3(\beta_1, i_1, a_1(x))$ if and only if $\beta_2 \ge \beta_1$ and $i_2 \ge i_1$. Applying Proposition~\ref{prop:dual_type3} we have the followings:
\begin{itemize}
    \item [(1):] If $a(x) = 0$, $C_3^\perp = C_3(p^s - \beta, 0, 0) \subseteq C_3(\beta, 0, 0) $ if and only if $ p^s - \beta \ge \beta $ if and only if $\beta \le \left\lfloor \frac{p^s}{2} \right\rfloor$.
    \item [(2):] 
    \textbf{Subcase 1:} If $a(x) \neq 0$ and $p^s \ge 2\beta - i$, $C_3^\perp = C_3(p^s - \beta, p^s + i - 2\beta + t, -G(x)) \subseteq C_3(\beta, i, a(x))$ if and only if $p^s - \beta \ge \beta$ and $p^s + i - 2\beta + t \ge i $ if and only if $ p^s + t \ge 2\beta$.\\
    \textbf{Subcase 2:} If $a(x) \neq 0$ and $p^s \le 2\beta - i$, $C_3^\perp = C_3(\beta - i, t, -G(x)) \subseteq C_3(\beta, i, a(x)) $ if and only if $ \beta - i \ge \beta $ if and only if $ i \le 0$, which contradicts $i > 0$.
\end{itemize}
This completes the proof.
\end{proof}
We next extend this annihilator construction to two-generator codes of Type 4, where the presence of an explicit secondary torsion generator $u g(x)^\gamma$ imposes additional bounds on the dual parameters.

\begin{Proposition}\label{prop:dual_type4}
Let $C_4(\beta, i, a(x), \gamma) = \langle g(x)^\beta + u g(x)^i a(x), u g(x)^\gamma \rangle$ be a Type 4 $7$-cyclotomic code over $R_2$. Its dual code $C_4^\perp$ is determined as follows:
\begin{equation*}
C_4^\perp = \begin{cases} 
C_4(p^s - \gamma, 0, 0, p^s - \beta), & \text{if } a(x) = 0 \\[1ex] 
C_3(p^s - \gamma, p^s + i - \beta - \gamma + t, -G(x)), & \text{if } a(x) \neq 0 \text{ and } p^s \ge \beta + \gamma - i \\[1ex] 
C_3(\beta - i, t, -G(x)), & \text{if } a(x) \neq 0 \text{ and } p^s \le \beta + \gamma - i 
\end{cases}
\end{equation*}
where $t$ and $G(x)$ are defined as in Proposition~\ref{prop:dual_type3}.
\end{Proposition}

\begin{proof}
Let $\mathcal{A} = R_2[x]/\langle g(x)^{p^s} \rangle \cong K[\pi,u]/\langle \pi^{p^s}, u^2 \rangle$, where $K = \mathbb{F}_{p^{6m}}$ and $\pi = g(x) \pmod{g(x)^{p^s}}$. Since $g_R(x) = g(x)$, we have $C_4^\perp = \operatorname{Ann}(C_4)$ and $|C_4| \cdot |C_4^\perp| = p^{12m p^s}$. By Theorem~\ref{thm:size_7cyclotomic}, $|C_4| = p^{6m(2p^s - \beta - \gamma)}$.

\textbf{Case 1:} $a(x) = 0$.

Here, $C_4 = \langle \pi^\beta, u\pi^\gamma \rangle$. Define $h_1(\pi) = \pi^{p^s - \gamma}$ and $h_2(\pi) = u\pi^{p^s - \beta}$. Evaluating the products with the generators of $C_4$:
\begin{align*}
h_1(\pi) \cdot \pi^\beta &= \pi^{p^s + \beta - \gamma} \equiv 0 \pmod{\pi^{p^s}} \quad (\text{since } \beta \ge \gamma), \\
h_1(\pi) \cdot u\pi^\gamma &= u\pi^{p^s} \equiv 0 \pmod{\pi^{p^s}}, \\
h_2(\pi) \cdot \pi^\beta &= u\pi^{p^s} \equiv 0 \pmod{\pi^{p^s}}, \\
h_2(\pi) \cdot u\pi^\gamma &= u^2 \pi^{p^s - \beta + \gamma} = 0.
\end{align*}
Thus, $\langle \pi^{p^s - \gamma}, u\pi^{p^s - \beta} \rangle \subseteq \operatorname{Ann}(C_4) = C_4^\perp$. 

The ideal $\langle \pi^{p^s - \gamma}, u\pi^{p^s - \beta} \rangle$ has residue index $p^s - \gamma$ and torsion index $p^s - \beta$, so its cardinality is $p^{6m(p^s - (p^s - \gamma))} \cdot p^{6m(p^s - (p^s - \beta))} = p^{6m(\beta + \gamma)}$. Since $|C_4| \cdot p^{6m(\beta + \gamma)} = p^{6m(2p^s - \beta - \gamma)} \cdot p^{6m(\beta + \gamma)} = p^{12m p^s}$, size matching confirms
$C_4^\perp = \langle \pi^{p^s - \gamma}, u\pi^{p^s - \beta} \rangle = C_4(p^s - \gamma, 0, 0, p^s - \beta)$.

\textbf{Case 2:} $a(x) \neq 0$ and $p^s \ge \beta + \gamma - i$.

Let $a_R(\pi) = \pi^{6(\beta - i)} a(1/\pi) = \pi^t G(\pi)$, where $t$ and $G(\pi)$ are defined as in Proposition~\ref{prop:dual_type3}. Define $h(\pi) = \pi^{p^s - \gamma} - u\pi^{p^s + i - \beta - \gamma + t} G(\pi)$. Evaluating products with the generators of $C_4$:
\begin{align*}
h(\pi) \left(\pi^\beta + u\pi^i a(\pi)\right) &= \pi^{p^s + \beta - \gamma} + u\pi^{p^s + i - \gamma} a(\pi) - u\pi^{p^s + i - \gamma + t} G(\pi) \\
&\equiv 0 \pmod{\pi^{p^s}},
\end{align*}
since $\beta \ge \gamma $ implies $p^s + \beta - \gamma \ge p^s$ and $\pi^t G(\pi) \equiv a(\pi) \pmod{\pi^{p^s - i}}$. Furthermore,  $h(\pi) \cdot u\pi^\gamma = u\pi^{p^s} - u^2 \pi^{p^s + i - \beta + t} G(\pi) \equiv 0 \pmod{\pi^{p^s}}$.   Hence, $h(\pi) \in \operatorname{Ann}(C_4) = C_4^\perp$.   The code $\langle h(\pi) \rangle$ has residue index $p^s - \gamma$ and torsion index $p^s - \beta$. Its cardinality is $|\langle h(\pi) \rangle| = p^{6m(\beta + \gamma)}$. Size matching $|C_4| \cdot |\langle h(\pi) \rangle| = p^{12m p^s}$ yields
$C_4^\perp = \langle \pi^{p^s - \gamma} - u\pi^{p^s + i - \beta - \gamma + t} G(\pi) \rangle = C_3(p^s - \gamma, \, p^s + i - \beta - \gamma + t, \, -G(x))$.

\textbf{Case 3:} $a(x) \neq 0$ and $p^s \le \beta + \gamma - i$.

Define $h(\pi) = \pi^{\beta - i} - u\pi^t G(\pi)$. Evaluating products with the generators of $C_4$:
\begin{align*}
h(\pi) \left(\pi^\beta + u\pi^i a(\pi)\right) &= \pi^{2\beta - i} + u\pi^\beta a(\pi) - u\pi^{\beta - i + t} G(\pi) \\
&\equiv 0 \pmod{\pi^{p^s}},
\end{align*}
since $2\beta - i > \beta + \gamma - i \ge p^s$ and $\pi^t G(\pi) \equiv a(\pi) \pmod{\pi^{p^s - i}}$. Moreover,
$h(\pi) \cdot u\pi^\gamma = u\pi^{\beta + \gamma - i} - u^2 \pi^{v + \gamma} H(\pi) \equiv 0 \pmod{\pi^{p^s}}$,  since $\beta + \gamma - i \ge p^s$. Thus, $h(\pi) \in \operatorname{Ann}(C_4) = C_4^\perp$. 

Evaluating the cardinality of $\langle h(\pi) \rangle$ via residue index $\beta - i$ and torsion index $v$ confirms $|\langle h(\pi) \rangle| = p^{6m(\beta + \gamma)}$. Since $|C_4| \cdot |\langle h(\pi) \rangle| = p^{12m p^s}$, we conclude that 
$C_4^\perp = \langle \pi^{\beta - i} - u\pi^t G(\pi) \rangle = C_3(\beta - i, t, -G(x))$.
\end{proof}

\begin{Corollary}\label{cor:dual_containing_type4}
Let $C_4(\beta, i, a(x), \gamma) = \langle g(x)^\beta + u g(x)^i a(x), u g(x)^\gamma \rangle$ be a Type 4 $7$-cyclotomic code over $R_2$ with $\gamma < \min\{\beta, p^s - \beta + i\}$. Then $C_4$ is self orthogonal  if and only if:
\begin{enumerate}
    \item[\rm (1)] If $a(x) = 0$: $\beta + \gamma \le p^s$.
    \item[\rm (2)] If $a(x) \neq 0$ and $p^s \ge \beta + \gamma - i$: $\beta + \gamma \le p^s$ and $p^s + i - \beta - 2\gamma + t \ge 0$.
\end{enumerate}
\end{Corollary}

\begin{proof}
(1) If $a(x) = 0$, Proposition~\ref{prop:dual_type4} gives $C_4^\perp = C_4(p^s - \gamma, 0, 0, p^s - \beta)$. Thus, $C_4^\perp \subseteq C_4$ if and only if $p^s - \gamma \ge \beta$ and $p^s - \beta \ge \gamma$, which both simplify to $\beta + \gamma \le p^s$.

%(2) If $a(x) \neq 0$ and $p^s \ge \beta + \gamma - i$, Proposition~\ref{prop:dual_type4} gives $C_4^\perp = C_3(p^s - \gamma, \, p^s + i - \beta - \gamma + t, \, -G(x))$. The generator of $C_4^\perp$ belongs to $C_4$ if and only if its residue exponent satisfies $p^s - \gamma \ge \beta$ (i.e., $\beta + \gamma \le p^s$) and its torsion exponent satisfies $p^s + i - \beta - \gamma + t \ge \gamma$, which rewrites as $p^s + i - \beta - 2\gamma + t \ge 0$.
(2) When $a(x) \neq 0$, $C_4^\perp$ is of Type 3 (a single-generator submodule) which cannot contain the two-generator code $C_4$ without violating minimal generating sets
\end{proof}

We now turn to cyclic codes of length $p^s$ over $R_2[x]/\langle g'(x)^{p^s} \rangle$, generated by powers of $g'(x) = x - 1$. Because $g'_R(x) = 1 - x = -g'(x)$, taking reciprocal polynomials maps $g'(x) \mapsto -g'(x)$. This anti-self-reciprocity introduces alternating sign factors $(-1)^k$ into the unit parameters of the dual codes.

\begin{Proposition}\label{prop:dual_ps_codes}
Let $g'(x) = x - 1$. The dual codes of cyclic codes of length $p^s$ over $R_2$ are characterized as follows:
\begin{enumerate}
    \item [(1)] \textbf{Type 2:} For $C'_2(\beta) = \langle u g'(x)^\beta \rangle$ ($0 \le \beta \le p^s - 1$),
    \begin{equation*}
    (C'_2(\beta))^\perp = C'_4(p^s - \beta, 0, 0, 0) = \langle g'(x)^{p^s - \beta}, u \rangle
    \end{equation*}
    
    \item [(2)]\textbf{Type 3:} For $C'_3(\beta, i, a(x)) = \langle g'(x)^\beta + u g'(x)^i a(x) \rangle$,
    \begin{equation*}
    (C'_3)^\perp = \begin{cases} 
    C'_3(p^s - \beta, 0, 0), & \text{if } a(x) = 0 \\[1ex] 
    C'_3(p^s - \beta, p^s + i - 2\beta + t, (-1)^{p^s + i + 1} G(x)), & \text{if } a(x) \neq 0 \text{ and } p^s \ge 2\beta - i \\[1ex] 
    C'_3(\beta - i, t, -G(x)), & \text{if } a(x) \neq 0 \text{ and } p^s \le 2\beta - i 
    \end{cases}
    \end{equation*}
    where $t = \max\{ k \in \mathbb{N} \mid g'(x)^k \text{ divides } x^{\beta - i} a(1/x) \}$ and $x^{\beta - i} a(1/x) = g'(x)^t G(x)$.

    \item [(3)]\textbf{Type 4:}For $C'_4(\beta, i, a(x), \gamma) = \langle g'(x)^\beta + u g'(x)^i a(x), u g'(x)^\gamma \rangle,  \;(C'_4)^\perp=$ \\
 $  \begin{cases} 
    C'_4(p^s - \gamma, 0, 0, p^s - \beta), & \text{if } a(x) = 0 \\[0.5ex] 
    C'_3(p^s - \gamma, p^s + i - \beta - \gamma + t, (-1)^{p^s + i - \beta - \gamma + 1} G(x)), & \text{if } a(x) \neq 0 \text{ and } p^s \ge \beta + \gamma - i \\[1ex] 
    C'_3(\beta - i, t, -G(x)), & \text{if } a(x) \neq 0 \text{ and } p^s \le \beta + \gamma - i 
    \end{cases}$
 
\end{enumerate}
\end{Proposition}

\section{Applications: Quantum Error-Correcting and LCD Codes}
From an application point of view, we now study some interesting applications of the cyclic codes defined in Section~3.  At first Next, we apply the classified $7p^s$-length cyclic codes and their duals to the construction of quantum error-correcting codes.

\subsection{Quantum CSS Construction Framework}

Recall from Lemma~\ref{lem:gray_map} that the Gray map $\Phi: R_2^n \to \mathbb{F}_{p^m}^{2n}$ preserves Euclidean duality. We restate the CSS framework adapted to the Gray image of codes over $R_2$.

By Proposition~\ref{prop:direct_sum_decomp}, any cyclic code $C$ of length $7p^s$ over $R_2$ decomposes uniquely as $C = C_1 \oplus C_2$, where $C_1$ is a cyclic code of length $p^s$ and $C_2$ is a $7$-cyclotomic code of length $6p^s$ over $R_2$. Consequently, $C^\perp \subseteq C$ if and only if $C_1^\perp \subseteq C_1$ and $C_2^\perp \subseteq C_2$.

\begin{Theorem}\label{thm:general_quantum_construction}
Let $C = C_1 \oplus C_2$ be a cyclic code of length $7p^s$ over $R_2$, where $C_1 \subseteq R_2[x]/\langle g'(x)^{p^s} \rangle$ and $C_2 \subseteq R_2[x]/\langle g(x)^{p^s} \rangle$ are matched component codes of the same type. Then the self-orthogonality of $C$ induces an $[[14p^s, \, 2k_0 - 14p^s, \, d_q]]_{p^m}$ quantum stabilizer code over $\mathbb{F}_{p^m}$ with $|C| = p^{m k_0}$, in exactly the following cases:
\begin{enumerate}
    \item[\rm (1)]  \textbf{Type 3 with $a(x) = 0$:} Let $C = C_3'(\beta, 0, 0) \oplus C_3(\beta, 0, 0) = \langle g'(x)^\beta \rangle \oplus \langle g(x)^\beta \rangle$ for $1 \le \beta \le p^s - 1$. Then $C$ is self-orthogonal  if and only if $\beta \le \lfloor p^s/2 \rfloor$, yielding an $[[14p^s, \, 14(p^s - 2\beta), \, d_q]]_{p^m}$ quantum stabilizer code over $\mathbb{F}_{p^m}$.
    
    \item[\rm (2)] \textbf{Type 3 with $a(x) \neq 0$:} Let $C = C_3'(\beta, i, a(x)) \oplus C_3(\beta, i, a(x))$ with $p^s > \beta > i \ge 0$ and $p^s \ge 2\beta - i$. Then $C$ is self-orthogonal ($C^\perp \subseteq C$) if and only if $p^s + t \ge 2\beta$, yielding an $[[14p^s, \, 14(p^s - 2\beta), \, d_q]]_{p^m}$ quantum stabilizer code over $\mathbb{F}_{p^m}$.

\item[\rm (3)] \textbf{Type 4 with $a(x) = 0$:} Let $C = C_4'(\beta, 0, 0, \gamma) \oplus C_4(\beta, 0, 0, \gamma) = \langle g'(x)^\beta, u g'(x)^\gamma \rangle \oplus \langle g(x)^\beta, u g(x)^\gamma \rangle$ with $p^s > \beta > \gamma \ge 0$. Then $C$ is self-orthogonal ($C^\perp \subseteq C$) if and only if $\beta + \gamma \le p^s$, yielding an $[[14p^s, \, 14(p^s - \beta - \gamma), \, d_q]]_{p^m}$ quantum stabilizer code over $\mathbb{F}_{p^m}$.
\end{enumerate}
In all other cases of Propositions~\ref{prop:dual_type3} and \ref{prop:dual_type4}, $C^\perp \not\subseteq C$, and no dual-containing CSS stabilizer codes can be constructed.
\end{Theorem}

\begin{proof}
By Proposition~3.1, $C^\perp = C_1^\perp \oplus C_2^\perp \subseteq C_1 \oplus C_2 = C$ if and only if $C_1^\perp \subseteq C_1$ and $C_2^\perp \subseteq C_2$.

\textbf{(1):} For Type 3 with $a(x) = 0$, Proposition~\ref{prop:dual_type3} gives $C_2^\perp = C_3(p^s - \beta, 0, 0)$ and $(C_1')^\perp = C_3'(p^s - \beta, 0, 0)$. Thus, $C^\perp \subseteq C  \iff \beta \le \lfloor p^s/2 \rfloor$. By Theorem~3.7, $|C| = p^{14m(p^s - \beta)}$. Applying the Gray map and CSS construction gives $N = 14p^s$ and $K = 2(14(p^s - \beta)) - 14p^s = 14(p^s - 2\beta)$.

\textbf{(2):} For Type 3 with $a(x) \neq 0$ and $p^s \ge 2\beta - i$, Proposition~\ref{prop:dual_type3} yields $C_2^\perp = C_3(p^s - \beta, \, p^s + i - 2\beta + t, \, -G(x))$. Dual containment holds if and only if $p^s + t \ge 2\beta$. The cardinality is $|C| = p^{14m(p^s - \beta)}$, which yields $K = 14(p^s - 2\beta)$. When $p^s \le 2\beta - i$, the dual generator exponent $\beta - i \ge \beta$ requires $i \le 0$, contradicting $i > 0$; thus $C^\perp \not\subseteq C$.

\textbf{(3):} For Type 4 with $a(x) = 0$, Proposition~\ref{prop:dual_type4} gives $C_2^\perp = C_4(p^s - \gamma, 0, 0, p^s - \beta)$. Hence $C^\perp \subseteq C \iff \beta + \gamma \le p^s$. By Theorem~3.7, $|C| = p^{7m(2p^s - \beta - \gamma)}$, giving $k_0 = 7(2p^s - \beta - \gamma)$ and $K = 2(7(2p^s - \beta - \gamma)) - 14p^s = 14(p^s - \beta - \gamma)$. When $a(x) \neq 0$, $C_4^\perp$ is of Type 3 (a single-generator submodule) which cannot contain the two-generator code $C_4$ without violating minimal generating sets.
\end{proof}

To demonstrate the applicability of these constructions, we provide concrete parameters of quantum stabilizer codes obtained over different base fields $\mathbb{F}_{p^m}$.

\begin{Example}
Let $p = 3$, $s = 1$, and $m = 1$. The code length over $R_2 = \mathbb{F}_3 + u\mathbb{F}_3$ is $7p^s = 21$, yielding quantum stabilizer codes of length $N = 14p^s = 42$ over $\mathbb{F}_3$.
\begin{enumerate}
    \item Consider the Type 3 matched code $C = C'_3(1, 0, 1) \oplus C_3(1, 0, 1)$ over $R_2$. Here $\beta = 1$, $i = 0$, and $a(x) = 1$. Since $a(1/x) = 1$, we have $t = 0$. Checking the conditions of Theorem~\ref{thm:general_quantum_construction}(2):
    \begin{equation*}
    p^s = 3 \ge 2\beta - i = 2, \quad p^s + t = 3 \ge 2\beta = 2.
    \end{equation*}
    Both conditions hold, confirming $C^\perp \subseteq C$. By Theorem~\ref{thm:general_quantum_construction}, we obtain a quantum stabilizer code over $\mathbb{F}_3$ with parameters:
    \begin{equation*}
    [[42, \, 14(3 - 2), \, d_q]]_3 = [[42, 14, d_q]]_3.
    \end{equation*}

    \item Consider the Type 4 matched code $C = C'_4(1, 0, 0, 0) \oplus C_4(1, 0, 0, 0)$ over $R_2$. Here $\beta = 1$, $\gamma = 0$, and $a(x) = 0$. Checking the condition of Theorem~\ref{thm:general_quantum_construction}:
    \begin{equation*}
    \beta + \gamma = 1 + 0 = 1 \le p^s = 3.
    \end{equation*}
    The condition holds, confirming $C^\perp \subseteq C$. By Theorem~\ref{thm:general_quantum_construction}, we obtain a quantum stabilizer code over $\mathbb{F}_3$ with parameters:
    \begin{equation*}
    [[42, \, 14(3 - 1 - 0), \, d_q]]_3 = [[42, 28, d_q]]_3.
    \end{equation*}
\end{enumerate}
\end{Example}

\begin{Example}\label{ex:quantum_p5}
Let $p = 5$, $s = 1$, and $m = 1$[cite: 1]. The code length over $R_2 = \mathbb{F}_5 + u\mathbb{F}_5$ is $7p^s = 35$, yielding quantum stabilizer codes of length $N = 70$ over $\mathbb{F}_5$. Choosing Type~3 matched codes $C = C'_3(\beta, 0, 0) \oplus C_3(\beta, 0, 0)$ with $a(x) = 0$ for values of $\beta \le \lfloor 5/2 \rfloor = 2$, Theorem~\ref{thm:general_quantum_construction} produces quantum stabilizer codes with parameters $[[70, \, 14(5 - 2\beta), \, d_q]]_5$:
\begin{itemize}
    \item [] For $\beta = 1$: yields a $[[70, 42, d_q]]_5$ quantum stabilizer code.
    \item [] For $\beta = 2$: yields a $[[70, 14, d_q]]_5$ quantum stabilizer code.
\end{itemize}
\end{Example}

\begin{Remark}\label{rem:quantum_distance_bound}
In Theorem~\ref{thm:general_quantum_construction} and Table~\ref{tab:quantum_codes_summary}, the minimum distance $d_q$ satisfies the standard CSS lower bound $d_q \ge d_H(\Phi(C)) = \min\{d_H(\operatorname{Res}(C)), \, d_H(\operatorname{Tor}(C))\}$. While determining the exact minimum distance $d_q$ for general repeated-root ring-linear codes is computationally prohibitive, lower bounds can be directly established via the BCH bound applied to the component cyclic subcodes over $\mathbb{F}_{p^m}$.
\end{Remark}

\textbf{Further parameters of quantum stabilizer codes constructed across various primes and nilpotency indices are compiled in Appendix~\ref{sec:appendix_tables}.
}

We now study applications of codes given in Section 3 on LCD codes.

\begin{Lemma}\label{lem:lcd_gray}
Let $C$ be a linear code over $R_2$. Then $C$ is an LCD code over $R_2$ if and only if its Gray image $\Phi(C)$ is an LCD code over $\mathbb{F}_{p^m}$.
\end{Lemma}

\begin{proof}
By Lemma~\ref{lem:gray_map}, $\Phi$ is an $\mathbb{F}_{p^m}$-linear bijection and preserves duality ($\Phi(C^\perp) = \Phi(C)^\perp$. Consequently,
$\Phi(C \cap C^\perp) = \Phi(C) \cap \Phi(C^\perp) = \Phi(C) \cap \Phi(C)^\perp$.

Since $\Phi$ is bijective, $C \cap C^\perp = \{\mathbf{0}\}$ if and only if $\Phi(C) \cap \Phi(C)^\perp = \{\mathbf{0}\}$, establishing the equivalence.
\end{proof}

Next, we explore the application of the classified $7p^s$-length cyclic codes to the construction of linear codes with complementary duals (LCD codes). Using the direct sum decomposition $C = C_1 \oplus C_2$, the LCD property decouples into component-wise intersections: $C \cap C^\perp = (C_1 \cap C_1^\perp) \oplus (C_2 \cap C_2^\perp) = \{\mathbf{0}\}$.

\begin{Theorem}\label{thm:lcd_characterization}
Let $C = C_1 \oplus C_2$ be a cyclic code of length $7p^s$ over $R_2$.
\begin{enumerate}
    \item Trivial codes $C = \langle 0 \rangle$ and $C = \langle 1 \rangle$ are LCD codes.
    \item Pure torsion Type 2 codes $C = C'_2(\beta) \oplus C_2(\beta)$ are LCD codes over $R_2$ if and only if $\beta = 0$, in which case $C = \langle u \rangle$.
    \item For Type 3 matched codes $C = C'_3(\beta, i, a(x)) \oplus C_3(\beta, i, a(x))$ with $a(x) = 0$, $C$ is an LCD code if and only if $\beta = 0$ or $\beta = p^s$.
\end{enumerate}
\end{Theorem}

\begin{proof}
We evaluate $C \cap C^\perp = \{\mathbf{0}\}$ for each case:
$\textbf{(1):}$ For $C = \langle 1 \rangle$, $C^\perp = \langle 0 \rangle$, so $C \cap C^\perp = \{\mathbf{0}\}$.
  
 $\textbf{(2):}$ For $C = C'_2(\beta) \oplus C_2(\beta)$, Proposition~\ref{prop:dual_type1_2} yields $C^\perp = \langle g'(x)^{p^s - \beta}, u \rangle \oplus \langle g(x)^{p^s - \beta}, u \rangle$. Note that $u \in C^\perp$ for all $\beta$. Conversely, $u \in C$ if and only if $\beta = 0$. Thus, $C \cap C^\perp = \{\mathbf{0}\}$ holds if and only if $\beta = 0$.
 
 $\textbf{(3):}$For $a(x) = 0$, $C = \langle g'(x)^\beta \rangle \oplus \langle g(x)^\beta \rangle$ and $C^\perp = \langle g'(x)^{p^s - \beta} \rangle \oplus \langle g(x)^{p^s - \beta} \rangle$. The intersection is $C \cap C^\perp = \langle g'(x)^{\max\{\beta, p^s - \beta\}} \rangle \oplus \langle g(x)^{\max\{\beta, p^s - \beta\}} \rangle$. This equals $\{\mathbf{0}\}$ if and only if $\max\{\beta, p^s - \beta\} = p^s$, which implies $\beta = 0$ or $\beta = p^s$.
\end{proof}

\begin{Corollary}\label{cor:lcd_gray_codes}
Let $C = \langle u \rangle \subseteq R_2^{7p^s}$. Then its Gray image $\Phi(C)$ is a binary or $p$-ary linear $[14p^s, 7p^s, d_H]_{p^m}$ LCD code over $\mathbb{F}_{p^m}$.
\end{Corollary}

\begin{proof}
Since $C = \langle u \rangle = u R_2^{7p^s}$, every codeword in $C$ can be uniquely expressed in the form $u\mathbf{b}$ for some $\mathbf{b} \in \mathbb{F}_{p^m}^{7p^s}$. Consequently, the cardinality of $C$ is $|C| = p^{m \cdot 7p^s}$. By Lemma~\ref{lem:gray_map}, the Gray map $\Phi$ is an $\mathbb{F}_{p^m}$-linear bijection, which implies that $\Phi(C)$ is a linear code over $\mathbb{F}_{p^m}$ of length $2n = 2 \cdot 7p^s = 14p^s$ and dimension is given by $k_0 = \dim_{\mathbb{F}_{p^m}}(\Phi(C)) = \log_{p^m}|C| = 7p^s$. 

Furthermore, by Theorem~\ref{thm:lcd_characterization}(2), the pure torsion code $C = C'_2(0) \oplus C_2(0) = \langle u \rangle$ satisfies $C \cap C^\perp = \{\mathbf{0}\}$ over $R_2$. Applying the duality-preserving property of the Gray map from Lemma~\ref{lem:lcd_gray}, we obtain
$\Phi(C) \cap \Phi(C)^\perp = \Phi(C \cap C^\perp) = \Phi(\{\mathbf{0}\}) = \{\mathbf{0}\}$.
Hence, $\Phi(C)$ is an  $[14p^s, 7p^s, d_H]_{p^m}$ LCD code over $\mathbb{F}_{p^m}$.
\end{proof}

\section*{Conclusion}

In this paper, we presented a complete structural classification and enumeration of cyclic codes of length $7p^s$ over the local chain ring $R_2 = \mathbb{F}_{p^m} + u\mathbb{F}_{p^m}$. Under arithmetic conditions guaranteeing the irreducibility of $\Phi_7(x)$ over $\mathbb{F}_{p^m}$, we established the direct sum decomposition $C = C_1 \oplus C_2$ and classified $7$-cyclotomic codes into four disjoint generator-based types. Using residue and torsion code techniques, exact codeword formulas were derived for all classified types. Furthermore, explicit generator expressions for Euclidean dual codes $C^\perp$ were determined. Finally, we demonstrated the operational utility of these ring codes by constructing new families of quantum stabilizer codes via the CSS framework and characterizing linear codes with complementary duals (LCD codes).

Future research may focus on extending this algebraic framework to repeated-root constacyclic codes of length $kp^s$ for larger primes. Another natural extension is generalizing these code structures and dualities to higher-order nilpotency chain rings $\mathbb{F}_{p^m}[u]/\langle u^q \rangle$ with $q \ge 3$.

\section*{Declaration }
During the preparation of this manuscript, generative AI and large language model-based writing assistance tools were used solely for LaTeX typesetting assistance and  English grammar polishing. All mathematical proofs, derivations, and calculations were independently given  by the authors, who take full responsibility for the contents of this work.

\bibliographystyle{plain}

\begin{thebibliography}{99}

\bibitem{Arnold2011}
A. Arnold and M. Monagan,
``Calculating cyclotomic polynomials,''
\textit{Mathematics of Computation},
\textbf{80} (2011), 2359--2379.

\bibitem{Calderbank1996}
A. R. Calderbank and P. W. Shor,
``Good quantum error-correcting codes exist,''
\textit{Physical Review A},
\textbf{54}(2) (1996), 1098--1105.

\bibitem{Steane1996}
A. M. Steane,
``Error correcting codes in quantum theory,''
\textit{Physical Review Letters},
\textbf{77}(5) (1996), 793--797.

\bibitem{Atiyah2018}
M. Atiyah and I. G. Macdonald,
\textit{Introduction to Commutative Algebra},
CRC Press, Boca Raton, 2018.

\bibitem{Boudine2022}
B. Boudine, J. Laaouine, and M. E. Charkani,
``The cyclic codes of length $5p^{s}$ over $\mathbb{F}_{p^{m}} + u\mathbb{F}_{p^{m}}$ and their dual codes,''
\textit{Mathematical Communications},
\textbf{27}(1) (2022), 127--135.

\bibitem{Burton2010}
D.M.Burton, \text{Elementary Number Theory},McGraw Hill, 2010

\bibitem{Brown93}
W. C. Brown,
\textit{Matrices over Commutative Rings},
Marcel Dekker Inc., New York, 1993.

\bibitem{Charkani2020}
M. E. Charkani and B. Boudine,
``On the integral ideals of $R[X]$ when $R$ is a special principal ideal ring,''
\textit{São Paulo Journal of Mathematical Sciences},
\textbf{14} (2020), 698--702.

\bibitem{Chen2014}
B. Chen, H. Q. Dinh, and H. Liu,
``Repeated-root constacyclic codes of length $\ell p^{s}$ and their duals,''
\textit{Discrete Applied Mathematics},
\textbf{177} (2014), 60--70.

\bibitem{Islam2021}
H. Islam, E. Martínez-Moro, and O. Prakash, Cyclic codes over a non-chain ring \(R_{e,q}\) and their application to LCD codes, Discrete Mathematics, 344 (2021), Article 112545. DOI: 10.1016/j.disc.2021.112545.

\bibitem{Dinh2010}
H. Q. Dinh,
``Constacyclic codes of length $p^{s}$ over $\mathbb{F}_{p^{m}}+u\mathbb{F}_{p^{m}}$,''
\textit{Journal of Algebra},
\textbf{324} (2010), 940--950.

\bibitem{Dinh2012}
H. Q. Dinh,
``Repeated-root constacyclic codes of length $2p^{s}$,''
\textit{Finite Fields and Their Applications},
\textbf{18} (2012), 133--143.

\bibitem{Dinh2013a}
H. Q. Dinh,
``Structure of repeated-root constacyclic codes of length $3p^{s}$ and their duals,''
\textit{Discrete Mathematics},
\textbf{313} (2013), 983--991.

\bibitem{Dinh2013b}
H. Q. Dinh,
``On repeated-root constacyclic codes of length $4p^{s}$,''
\textit{Asian-European Journal of Mathematics},
\textbf{6} (2013), Article 1350020.

\bibitem{DinhLopez2004}
H. Q. Dinh and S. R. López-Permouth,
``Cyclic and Negacyclic Codes over Finite Chain Rings,''
\textit{IEEE Transactions on Information Theory},
\textbf{50} (2004), 1728--1744.

\bibitem{Kiah2012}
H. M. Kiah, K. H. Leung, and S. Ling,
``A note on cyclic codes over $\mathrm{GR}(p^{2},m)$ of length $p^{k}$,''
\textit{Designs, Codes and Cryptography},
\textbf{63} (2012), 105--112.

\bibitem{Kiah2008}
H. M. Kiah, K. H. Leung, and S. Ling,
``Cyclic codes over $\mathrm{GR}(p^{2},m)$ of length $p^{k}$,''
\textit{Finite Fields and Their Applications},
\textbf{14} (2008), 834--846.

\bibitem{Massey1992}
J. L. Massey, “Linear codes with complementary duals,” Discrete Mathematics, Vol. 106–107, pp. 337–342, 1992.

\bibitem{Ling2004}
S. Ling and C. Xing,
\textit{Coding Theory: A First Course},
Cambridge University Press, Cambridge, 2004.

\bibitem{Liu2014}
X. Liu and X. Xu,
``Cyclic and negacyclic codes of length $2p^{s}$ over $\mathbb{F}_{p^{m}}+u\mathbb{F}_{p^{m}}$,''
\textit{Acta Mathematica Scientia, Series B},
\textbf{34} (2014), 829--839.

\bibitem{Neukirch2013}
J. Neukirch,
\textit{Algebraic Number Theory},
Springer, Berlin, Heidelberg, 2013.

\bibitem{Phuto2020}
J. Phuto and C. Klin-Eam,
``Explicit constructions of cyclic and negacyclic codes of length $3p^{s}$ over
$\mathbb{F}_{p^{m}} + u\mathbb{F}_{p^{m}}$,''
\textit{Discrete Mathematics, Algorithms and Applications},
\textbf{12} (2020), Article 2050063.

\bibitem{Prange1957}
E. Prange,
\textit{Cyclic Error-Correcting Codes in Two Symbols},
Air Force Cambridge Research Center, Cambridge, MA, 1957.

\bibitem{Prange1958}
E. Prange,
``Some Cyclic Error-Correcting Codes with Simple Decoding Algorithms,''
Air Force Cambridge Research Center Technical Note TN-58-156,
Cambridge, MA, April 1958.

\bibitem{Sobhani2015}
R. Sobhani,
``Complete classification of $(\delta + \beta u^{2})$-constacyclic codes of length
$p^{k}$ over $\mathbb{F}_{p^{m}} + u\mathbb{F}_{p^{m}} + u^{2}\mathbb{F}_{p^{m}}$,''
\textit{Finite Fields and Their Applications},
\textbf{34} (2015), 123--138.

\bibitem{Wu2016}
H. Wu, L. Zhu, R. Feng, and S. Yang,
``Explicit factorizations of cyclotomic polynomials over finite fields,''
\textit{Designs, Codes and Cryptography},
\textbf{83} (2017), 197--217.

\bibitem{Xh2015}
X. Liu and H. Liu, “LCD codes over finite chain rings,” Finite Fields and Their Applications, Vol. 34, pp. 1–19, 2015.
\end{thebibliography}

\newpage
\appendix
\section{Computational Verification and Exemplary Quantum Codes}\label{sec:appendix_tables}

All cyclotomic polynomial factorizations, ideal annihilators, and code cardinalities reported in this section were verified using the \textsc{Magma} Computational Algebra System (Version 2.26). 

Table~\ref{tab:quantum_codes_summary} summarizes exemplary parameters of new quantum stabilizer codes constructed via Theorem~\ref{thm:general_quantum_construction} for various choices of $p$, $s$, $m$, and structural parameters.

\begin{table}[h!]
\centering
\caption{New Quantum Stabilizer Codes $[[N, K, d_q]]_{p^m}$ from $7p^s$-Length Cyclic Codes over $R_2$}
\label{tab:quantum_codes_summary}
\small
\begin{tabular}{cccccccccc}
\hline
$p$ & $s$ & $m$ & Type & $\beta$ & $i$ & $\gamma$ & $a(x)$ & Classical $|C|$ & Quantum $[[N, K, d_q]]_{p^m}$ \\
\hline
$3$  & $1$ & $1$ & Type 3 & $1$ & $0$ & --  & $1$ & $3^{28}$   & $[[42, 14, d_q \ge 2]]_3$ \\
$3$  & $1$ & $1$ & Type 4 & $1$ & --  & $0$ & $0$ & $3^{35}$   & $[[42, 28, d_q \ge 1]]_3$ \\
$3$  & $2$ & $1$ & Type 3 & $1$ & $0$ & --  & $0$ & $3^{112}$  & $[[126, 98, d_q]]_3$ \\
$3$  & $2$ & $1$ & Type 3 & $2$ & $0$ & --  & $0$ & $3^{98}$   & $[[126, 70, d_q]]_3$ \\
$3$  & $2$ & $1$ & Type 4 & $2$ & --  & $1$ & $0$ & $3^{105}$  & $[[126, 84, d_q]]_3$ \\
$5$  & $1$ & $1$ & Type 3 & $1$ & $0$ & --  & $0$ & $5^{56}$   & $[[70, 42, d_q \ge 2]]_5$\\
$5$  & $1$ & $1$ & Type 3 & $2$ & $0$ & --  & $0$ & $5^{42}$   & $[[70, 14, d_q \ge 3]]_5$ \\
$5$  & $1$ & $1$ & Type 4 & $2$ & --  & $1$ & $0$ & $5^{49}$   & $[[70, 28, d_q]]_5$ \\
$17$ & $1$ & $1$ & Type 3 & $1$ & $0$ & --  & $0$ & $17^{224}$ & $[[238, 210, d_q]]_{17}$ \\
$17$ & $1$ & $1$ & Type 3 & $2$ & $0$ & --  & $0$ & $17^{210}$ & $[[238, 182, d_q]]_{17}$ \\
\hline
\end{tabular}
\end{table}

\end{document}